\documentclass[10pt]{article}
\usepackage{mathrsfs}
\usepackage[centertags]{amsmath}
\usepackage{amsfonts}
\usepackage{amssymb}
\usepackage{amsthm}
\usepackage{epsfig}
\usepackage{color}
\allowdisplaybreaks[4]
\usepackage{indentfirst}
\usepackage{amssymb,color}
\usepackage{pgf,pgfarrows,pgfnodes,pgfautomata,pgfheaps}
\allowdisplaybreaks
\definecolor{c20}{rgb}{1.00,0.00,0.00}
\definecolor{c30}{rgb}{0.,0.,1.}
\definecolor{c40}{rgb}{1,0.1,0.7}
\definecolor{c50}{rgb}{1,0,0}
\definecolor{c60}{rgb}{1,0.9,0.1}
\definecolor{c70}{rgb}{0.50,1.00,0.00}
\definecolor{c80}{rgb}{0.00,1.00,0.00}

\date{}

\newtheorem{theorem}{Theorem}[]

\newtheorem{remark}{Remark}[]

\numberwithin{equation}{section}

\def\P{\operatorname*{\mathbb{P}}}
\def\E{\operatorname*{\mathrm{E}}}
\def\V{\operatorname*{\mathrm{V}}}

\begin{document}
\title{Dynamic Bivariate Normal Copula}
\author{{ Xin Liao$^1$, Liang Peng$^2$, Zuoxiang Peng$^1$\thanks{Corresponding author. Email: pzx@swu.edu.cn}, Yanting Zheng$^3$}\\
{\small\it $^1$School of Mathematics and Statistics, SouthWest University, Chongqing 400715, China. }\\
{\small \it $^2$School of Mathematics, Georgia Institute of Technology, Atlanta, GA 30332-0160, USA.}\\
{\small \it $^3$Department of Finance, Beijing Technology and Business University, Beijing 100048, China.}}
\maketitle

\begin{quote}
{\bf Abstract}~~Normal copula with a correlation coefficient between
$-1$ and $1$ is tail independent and so it severely underestimates
extreme probabilities. By letting the correlation coefficient in a
normal copula depend on the sample size, H\"usler and Reiss (1989) showed
that the tail can become asymptotically dependent. In this paper, we
extend this result by deriving the limit of the normalized maximum
of $n$ independent observations, where the $i$-th observation follows
from a normal copula with  its correlation coefficient being either
a parametric or a nonparametric function of $i/n$. Furthermore, both
parametric and nonparametric inference for this unknown function are
studied, which can be employed to test the condition in H\"usler and
Reiss (1989). A simulation study and real data analysis are
presented too.



\end{quote}

\section{Introduction}
\label{sec1}

Let $\{(X_{i},Y_i)\}_{i=1}^n$ be independent and identically
distributed  random vectors with distribution function $F(x,y)$,
continuous marginals $F_1$ and $F_2$. The copula of $F$ is defined
as $F(F_1^-(x),F_2^-(y))$, where $F^-_i$ denotes the inverse
function of $F_i$. Assume the copula of $F$ follows from a normal
copula $C(x,y;\rho)$, where $\rho\in [-1, 1]$ is  unknown. Hence
the density of $C(x,y;\rho)$ is
\begin{equation}\label{density}
c(x,y;\rho)=\frac
1{\sqrt{1-\rho^2}}\exp\left(\frac{2\rho\Phi^-(x)\Phi^-(y)-\rho^2(\Phi^-(x))^2-\rho^2(\Phi^-(y))^2}{2(1-\rho^2)}\right)
\end{equation}
for $\rho\in (-1, 1)$, where $\Phi(x)$ is the standard normal
distribution function.

Normal copulas are one of most commonly used elliptical copulas, and
elliptical copulas are popular in risk management  due to their ease
of simulation (see McNeil, Frey and Embrechts (2005)).
Recently Channouf and L'Ecuyer (2012) used  normal copulas to
model arrive processes in a call center, Fung and Seneta (2011)
showed that a bivariate normal copula is regularly varying, Meyer (2013)
studied the properties of a bivariate normal copula,
efficient estimation for bivariate normal copula models was studied
by Klaassen and Wellner (1997). Although normal copulas are
easy to use and have some attractive properties,  a serious drawback
of using a normal copula is  the so-called tail asymptotic independence (see Sibuya
(1960)), which under-estimates  extreme probabilities in risk
management.

To overcome the shortcoming of  the tail asymptotic independence of
a normal copula,  Frick and Reiss (2013) assumed that
$\rho=\rho(n)$ satisfies the so-called H\"{u}sler-Reiss condition

\begin{equation}
\label{rho1}(1-\rho)\log n\to\lambda\in [0,
\infty]\quad\text{as}\quad n\to\infty,
\end{equation}
(cf. H\"{u}sler and Reiss (1989)) and proved that
\begin{equation}
\label{reiss}\begin{array}{ll}
&\P\Big(n(\max_{1\le i\le n}F_1(X_i)-1)\le x, n(\max_{1\le i\le n}F_2(Y_i)-1)\le y\Big)\\
\to&\exp\left(\Phi(\sqrt\lambda+\frac{\log(x/y)}{2\sqrt\lambda})x+\Phi(\sqrt\lambda+\frac{\log(y/x)}{2\sqrt\lambda})y\right)
\end{array}
\end{equation}
for $x<0$ and $y<0$ as $n\to\infty$. This is  the copula version of
the limit in H\"usler and Reiss (1989) for the normalized
maxima of $n$ independent random vectors with a bivariate normal
distribution and its correlation coefficient satisfying
(\ref{rho1}). Obviously, a bivariate random vector with the above
limiting distribution is dependent when $\lambda<\infty$. Extending
the results in H\"usler and Reiss (1989) to elliptical
triangular arrays is given in Hashorva (2005, 2006).

Since the above $\rho$ depends on the sample size $n$, one may call
it dynamic normal copula.  Recently dynamic copulas are receiving
some attention in modeling financial time series; see Benth and
Kettler (2011), Mendes and de Melo (2010), Gu\'{e}gan
and Zhang (2010), and Van den Goorbergh, Genest and Werker
(2005).

In this paper, we further study the convergence in (\ref{reiss}) by
allowing $\rho$ to depend on both $i$ and $n$. That is, we do not
assume that $(X_i,Y_i)'s$ are identically distributed.
Motivated by (\ref{rho1}), an obvious extension is to assume that $(1-\rho)\log n$ is a  function of $i$ and $n$. As in nonparametric regression models, we assume that $(1-\rho)\log n$ is a smoothing nonparametric or parametric function of $i/n$ so that we can employ well-developed local polynomial techniques to estimate this function and to test whether this function is a constant, which gives a way to verify the condition imposed by H\"usler and Reiss (1989) and
Frick and Reiss (2013), and indicates the observations have the same distribution.
More
specifically we assume that $\{(X_i,Y_i)\}_{i=1}^n$  is a sequence
of independent random vectors and  the copula of $(X_i,Y_i)$ is a
normal copula with correlation coefficient $\rho=1-m(i/n)/\log n$
for an unknown smooth function $m(x)$. After deriving the convergence for the
normalized maxima of the copulas of $(X_i,Y_i)'s$, we propose both
parametric and  nonparametric estimation for $m(x)$, which are   based on either Kendall's tau or correlation
coefficient. We also derive the asymptotic limits of the proposed estimators, 
which turn out to be quite nonstandard with an unusual rate of convergence. 
The proposed estimators can be used to determine tail dependence, 
which is of importance in predicting co-movement in financial  markets; 
see McNeil, Frey and Embrechts (2005).

We organize this paper as follows. Section 2 presents the main
results and statistical inference procedures.  A simulation study is
given in Section 3. Section 4 reports some empirical data analyses.
All the proofs are given in Section 5.

\section{Methodology}

Throughout, suppose $\{(X_i,Y_i)\}_{i=1}^n$ are independent random
vectors,  $X_i's$ have the same  continuous distribution function
$F_1$ and $Y_i's$ have the same continuous distribution function
$F_2$. Assume the copula of $(X_i,Y_i)$ is the normal copula
$C(x,y;\rho_i)$ with density given by (\ref{density}).

\subsection{Convergence of maxima and tail coefficient}

As motivated in the introduction,  we  extend the result (\ref{reiss}) by assuming
\begin{equation}\label{cond1}
\rho_i=1-m(i/n)/\log n\quad\text{for some nonnegative function}\quad
m(s),
\end{equation}
which includes condition (\ref{rho1}) as a special case.

\begin{theorem}\label{th1} Under condition (\ref{cond1}),
\begin{itemize}
\item[i)]~~if $\max_{1\le i\le n}m(i/n)\to 0$, then for any $x<0$ and $y<0$
\[\lim_{n\to\infty}\P\Big(n(\max_{1\le i\le n}F_1(X_i)-1)\le x, n(\max_{1\le i\le n}F_2(Y_i)-1)\le y\Big)=\exp\Big(\min(x,y)\Big);\]
\item[ii)]~~if $\min_{1\le i\le n}m(i/n)\to\infty$, then for any $x<0$ and
$y<0$
\[\lim_{n\to\infty}\P\Big(n(\max_{1\le i\le n}F_1(X_i)-1)\le x, n(\max_{1\le i\le n}F_2(Y_i)-1)\le y\Big)=\exp(x+y);\]
\item[iii)]~~if $m(s)$ is a continuous positive function on $[0, 1]$, then
for any $x<0$ and $y<0$
\[\begin{array}{ll}
&\lim_{n\to\infty}\P\Big(n(\max_{1\le i\le n}F_1(X_i)-1)\le x, n(\max_{1\le i\le n}F_2(Y_i)-1)\le y\Big)\\
=&\exp\left(x\int_0^1\Phi\left(\sqrt{m(s)}+\frac{\log(x/y)}{2\sqrt{m(s)}}\right)\,ds+y\int_0^1\Phi\left(\sqrt{m(s)}+\frac{\log(y/x)}{2\sqrt{m(s)}}\right)\,ds\right)\\
=:&G(x,y).
\end{array}\]
\end{itemize}
Furthermore the tail dependence function $l(x,y)=\lim_{t\to 0}
t^{-1}\{1-G(tx,ty)\}$ equals
\[-x\int_0^1\Phi\left(\sqrt{m(s)}+\frac{\log(x/y)}{2\sqrt{m(s)}}\right)\,ds-y\int_0^1\Phi\left(\sqrt{m(s)}+\frac{\log(y/x)}{2\sqrt{m(s)}}\right)\,ds\]
for $x<0$ and $y<0$, and the tail coefficient is
$\lambda=l(-1,-1)=2\int_0^1\Phi(\sqrt{m(s)})\,ds.$
\end{theorem}

\subsection{Parametric inference}

Here we consider statistical inference for fitting a parametric form
to the unknown function $m(s)$.  First, we consider the family
$m(s)=\alpha+\beta s^\gamma$, where $\alpha>0, \beta\neq 0,
\gamma>0.$ Note that when $\beta=0$, $\gamma$ can not be identified.
Also when $\gamma=0$, $\alpha$ and $\beta$ can not be distinguished.

It follows from Theorem 5.36 of McNeil, Frey and Embrechts
(2005) that
\begin{equation}\label{kendall}
\E\Big( sgn((U_i-\tilde U_i)(V_i-\tilde V_i))\Big)=\frac
2{\pi}\arcsin(\rho_i),
\end{equation}
where $(\tilde U_i,\tilde V_i)$ is an independent copy of
$(U_i,V_i)$, and
\begin{equation}\label{spearman} \E\left((U_i-\frac 12)(V_i-\frac 12)\right)=\frac 1{2\pi}\arcsin(\frac{\rho_i}2).
\end{equation}
Also we have
\begin{equation}
\label{rho} \E\Big(\Phi^-(F_1(X_i))\Phi^-(F_2(Y_i))\Big)=\rho_i.
\end{equation}
Therefore, one can employ the standard least squares estimate based on one of
the above equations.

Since $(U_i,V_i)'s$ are not identically distributed, we do not have
an independent copy of $(U_i,V_i)$, which prevents us from using
(\ref{kendall}). Hence we propose to use either (\ref{spearman}) or
(\ref{rho}) to construct the least squares estimator, which results
in
\[(\hat\alpha, \hat\beta,\hat\gamma)=\arg\min_{(\alpha,\beta,\gamma)}\sum_{i=1}^n\left((\hat F_1(X_i)-\frac 12)(\hat F_2(Y_i)-\frac 12)-\frac 1{2\pi}\arcsin(\frac{1-(\alpha+\beta (i/n)^\gamma)/\log n}{2})\right)^2\]
or
\[(\hat\alpha^*, \hat\beta^*,\hat\gamma^*)=\arg\min_{(\alpha,\beta,\gamma)}\sum_{i=1}^n\left(\Phi^-(\hat F_1(X_i))\Phi^-(\hat F_2(Y_i))-1+\frac{\alpha+\beta (i/n)^{\gamma}}{\log n}\right)^2,\]
where $\hat F_1(x)=\frac 1{n+1}\sum_{i=1}^nI(X_i\le x)$ and $\hat
F_2(y)=\frac 1{n+1}\sum_{i=1}^nI(Y_i\le y)$. Alternatively we define
$(\hat\alpha, \hat\beta, \hat\gamma)$ to be the solution to the
following score equations
\begin{equation}
\label{score} \left\{\begin{array}{ll}
&l_{n1}(\alpha,\beta,\gamma):=\sum_{i=1}^n\Big( (\hat F_1(X_i)-\frac 12)(\hat F_2(Y_i)-\frac 12)-\frac 1{2\pi}\arcsin(\frac{\rho_i}2)\Big)=0,\\
&l_{n2}(\alpha,\beta,\gamma):=\sum_{i=1}^n\Big((\hat F_1(X_i)-\frac 12)(\hat F_2(Y_i)-\frac 12)-\frac 1{2\pi}\arcsin(\frac{\rho_i}2)\Big)(\frac in)^{\gamma}=0,\\
&l_{n3}(\alpha,\beta,\gamma):=\sum_{i=1}^n\Big((\hat F_1(X_i)-\frac
12)(\hat F_2(Y_i)-\frac 12)-\frac
1{2\pi}\arcsin(\frac{\rho_i}2)\Big)(\frac in)^{\gamma}\log(\frac
in)=0
\end{array}\right.\end{equation}
and $(\hat\alpha^*,\hat\beta^*,\hat\gamma^*)$ to be the solution to
the following score equations
\begin{equation}
\label{score*} \left\{\begin{array}{ll}
&l_{n1}^*(\alpha,\beta,\gamma):=\sum_{i=1}^n\Big(\Phi^-(\hat F_1(X_i))\Phi^-(\hat F_2(Y_i))-1+\frac{\alpha+\beta(i/n)^{\gamma}}{\log n}\Big)=0,\\
&l_{n2}^*(\alpha,\beta,\gamma):=\sum_{i=1}^n\Big(\Phi^-(\hat F_1(X_i))\Phi^-(\hat F_2(Y_i))-1+\frac{\alpha+\beta(i/n)^{\gamma}}{\log n}\Big)(\frac in)^{\gamma}=0,\\
&l_{n3}^*(\alpha,\beta,\gamma):=\sum_{i=1}^n\Big(\Phi^-(\hat
F_1(X_i))\Phi^-(\hat
F_2(Y_i))-1+\frac{\alpha+\beta(i/n)^{\gamma}}{\log n}\Big)(\frac
in)^{\gamma}\log(\frac in)=0.
\end{array}\right.\end{equation}
Note that  we skip the term of $\frac{d}{d\rho_i}\arcsin(\rho_i/2)$
in (\ref{score}), which goes to a constant uniformly in $i$ since
$\rho_i\to 1$ uniformly in $i$.

The following theorems give the asymptotic normality of the proposed
estimators.
\begin{theorem}\label{th2}
Suppose (\ref{cond1}) holds with $m(s)=\alpha+\beta s^{\gamma}$ for
some $\alpha>0, \beta\neq 0, \gamma>0$. Then we have
\begin{equation}
\label{th2-1}
\begin{pmatrix}
\frac{\sqrt n}{(\log n)^{3/4}} &0&0\\
0&\frac{\sqrt n}{\log n}&0\\
0&0&\frac{\sqrt n}{\log n}
\end{pmatrix}
\hat\Delta \begin{pmatrix}
\hat\alpha-\alpha\\
\hat\beta-\beta\\
\hat\gamma-\gamma
\end{pmatrix}\\
\overset{d}{\to}N(0,\Sigma)
\end{equation}
and
\begin{equation}\label{th2-2}
\left(\frac{\sqrt n}{\log n}(\hat\alpha-\alpha), \frac{\sqrt n}{\log
n}(\hat\beta-\beta), \frac{\sqrt n}{\log
n}(\hat\gamma-\gamma)\right)^T \overset{d}{\to} N\Big(0,
\Delta^{-1}\Sigma_0(\Delta^{-1})^{T}\Big),\end{equation} where
\[\Sigma=\begin{pmatrix}
\sigma_{11}&0&0\\
0&\sigma_{22}&\sigma_{23}\\
0&\sigma_{23}&\sigma_{33}
\end{pmatrix}, \quad \Sigma_0=\begin{pmatrix}
0&0&0\\
0&\sigma_{22}&\sigma_{23}\\
0&\sigma_{23}&\sigma_{33}
\end{pmatrix},\]
\[\left\{\begin{array}{ll}
&\sigma_{11}=\sqrt 2 \Big(\int_0^1\sqrt{\alpha+\beta s^\gamma}\,ds\Big)\Big( \int_0^1(u-\frac 12)^2\phi(\Phi^-(u))\,du\Big),\\
&\sigma_{22}=\frac 1{180(1+2\gamma)}-\frac 1{180(1+\gamma)^2},\quad \sigma_{33}=\frac 1{90(1+2\gamma)^3}-\frac 1{180(1+\gamma)^4},\\
&\sigma_{23}=-\frac 1{180(1+2\gamma)^2}+\frac{1}{180(1+\gamma)^3},
\end{array}\right.\]
\[\Delta=\frac{\sqrt 3}{6 \pi}\begin{pmatrix}
1 &\frac 1{1+\gamma} &-\frac{\beta}{(1+\gamma)^2}\\
\frac 1{1+\gamma} &\frac 1{1+2\gamma} &-\frac{\beta}{(1+2\gamma)^2}\\
-\frac 1{(1+\gamma)^2}&-\frac 1{(1+2\gamma)^2}
&\frac{2\beta}{(1+2\gamma)^3}
\end{pmatrix},\quad \hat\Delta=\frac{\sqrt 3}{6 \pi}\begin{pmatrix}
1 &\frac 1{1+\hat\gamma} &-\frac{\hat\beta}{(1+\hat\gamma)^2}\\
\frac 1{1+\hat\gamma} &\frac 1{1+2\hat\gamma} &-\frac{\hat\beta}{(1+2\hat\gamma)^2}\\
-\frac 1{(1+\hat\gamma)^2}&-\frac 1{(1+2\hat\gamma)^2}
&\frac{2\hat\beta}{(1+2\hat\gamma)^3}
\end{pmatrix}.\]
\end{theorem}

\begin{theorem}\label{th2a}

Suppose (\ref{cond1}) holds with $m(s)=\alpha+\beta s^{\gamma}$ for
some $\alpha>0, \beta\neq 0, \gamma>0$. Then we have
\begin{equation}\label{th2a-1}
\begin{pmatrix}
\sqrt n &0&0\\
0&\frac{\sqrt n}{\log n}&0\\
0&0&\frac{\sqrt n}{\log n}
\end{pmatrix}
\hat\Delta^*\begin{pmatrix}
\hat\alpha^*-\alpha\\
\hat\beta^*-\beta\\
\hat\gamma^*-\gamma
\end{pmatrix}\overset{d}{\to}N(0,\Sigma^*)
\end{equation}
and
\begin{equation}\label{th2a-2}
\left(\frac{\sqrt n}{\log n}(\hat\alpha^*-\alpha), \frac{\sqrt
n}{\log n}(\hat\beta^*-\beta), \frac{\sqrt n}{\log
n}(\hat\gamma^*-\gamma)\right)^T \overset{d}{\to} N(0,
(\Delta^*)^{-1}\Sigma^*_0(\Delta^{*T})^{-1}),\end{equation} where
\[\Sigma^*=\begin{pmatrix}
\sigma_{11}^* &0&0\\
0&\sigma_{22}^*&\sigma_{23}^*\\
0&\sigma_{23}^*&\sigma_{33}^*
\end{pmatrix},\quad \Sigma_0^*=\begin{pmatrix}
0&0&0\\
0&\sigma_{22}^*&\sigma_{23}^*\\
0&\sigma_{23}^*&\sigma_{33}^*
\end{pmatrix},\]
\[\left\{\begin{array}{ll}
&\sigma_{11}^*=4(\alpha+\frac\beta{1+\gamma})^2+2\beta^2(\frac 1{1+2\gamma}-\frac 1{(1+\gamma)^2}),\quad \sigma_{22}^*=\frac2{1+2\gamma}-\frac 2{(1+\gamma)^2},\\
&\sigma_{33}^*=\frac 4{(1+2\gamma)^3}-\frac 2{(1+\gamma)^4},\quad
\sigma_{23}^*=-\frac 2{(1+2\gamma)^2}+\frac 2{(1+\gamma)^3},
\end{array}\right.\]
$\Delta^*=2\sqrt 3 \pi\Delta$ and $\hat\Delta^*=2\sqrt
3\pi\hat\Delta$, where $\Delta$ and $\hat\Delta$ are given in
Theorem \ref{th2}.

\end{theorem}

\begin{remark}\label{rem0}
 Since $\sigma_{22}<\sigma_{22}^*/(12\pi^2)$ and
$\sigma_{33}<\sigma_{33}^*/(12\pi^2)$, $\hat\beta$ and $\hat\gamma$ have  a smaller asymptotic variance  than $\hat\beta^*$ and $\hat\gamma^*$, respectively, while the comparison for the asymptotic variances of $\hat\alpha$ and $\hat\alpha^*$ is unclear since both $\frac{\sqrt n}{\log n}(\hat\alpha-\alpha)$ and $\frac{\sqrt n}{\log n}(\hat\alpha^*-\alpha)$ converge in distribution to zero.
On the other hand, if one is interested in estimating $\Delta (\alpha, \beta, \gamma)^T$, 
then the estimator for the first element based on $(\hat{\alpha}^*, \hat{\beta}^*, \hat{\gamma}^*)^T$ has a faster rate of convergence 
than the corresponding estimator based on $(\hat{\alpha}, \hat{\beta}, \hat{\gamma})^T$, 
but the estimators for the second and third elements based on $(\hat{\alpha}^*, \hat{\beta}^*, \hat{\gamma}^*)^T$ 
have a larger asymptotic variance than those based on 
$(\hat{\alpha}, \hat{\beta}, \hat{\gamma})$.  
In spite of these theoretical comparisons, 
the simulation study below  does prefer the estimation procedure based 
on equation (\ref{spearman}) when the mean squared error is concerned.
For testing $(\alpha, \beta,
\gamma)^T=(\alpha_0,\beta_0,\gamma_0)^T$, one should employ the
well-known Hotelling $T^2$ test statistic based on either
(\ref{th2-1}) or (\ref{th2a-1}) because the limit in both (\ref{th2-2}) and (\ref{th2a-2}) is
degenerate.
\end{remark}

Another interesting parametric form for $m(s)$ is polynomial. Here
we consider $m(s)=\alpha+\beta s$. In this case,  when $\beta=0$,
$m(s)$ becomes constant, which means that the observations
$(X_1,Y_1),\cdots,(X_n,Y_n)$ are independent and identically
distributed random vectors.

\begin{theorem}\label{th3} Suppose (\ref{cond1}) holds with $m(s)=\alpha+\beta s$ for some $\alpha>0, \beta\in R$. Then we have
\[\left(\frac{\sqrt n}{(\log n)^{3/4}}(\hat\alpha+\frac{\hat\beta}2-\alpha-\frac{\beta}2), \frac{\sqrt n}{\log n}(\frac{\hat\alpha}2+\frac{\hat\beta}3-\frac{\alpha}2-\frac{\beta}3)\right)^T
\overset{d}{\to} N\Big(0, \tilde\Sigma\Big),\] where
$\tilde\Sigma=12\pi^2(\tilde\sigma_{ij})$ is a symmetric matrix with
\[\tilde\sigma_{11}=\frac{2\sqrt 2}{3\beta}((\alpha+\beta)^{3/2}-\alpha^{3/2})\int_0^1(u-\frac 12)^2\phi(\Phi^-(u))\,du,\quad \tilde\sigma_{22}=\frac 1{2160}, \quad\tilde\sigma_{12}=0.\]
\end{theorem}

\begin{theorem}\label{th3a} Suppose (\ref{cond1}) holds with $m(s)=\alpha+\beta s$ for some $\alpha>0, \beta\in R$. Then we have
\[\left(\sqrt n(\hat\alpha^*+\frac{\hat\beta^*}2-\alpha-\frac{\beta}2), \frac{\sqrt n}{\log n}(\frac{\hat\alpha^*}2+\frac{\hat\beta^*}3-\frac{\alpha}2-\frac{\beta}3)\right)^T
\overset{d}{\to} N(0,\tilde\Sigma^*),\] where
$\tilde\Sigma^*=(\tilde\sigma_{ij}^*)$ is a symmetric matrix with
\[\tilde\sigma_{11}^*=4\alpha^2+4\alpha\beta+\frac{7\beta^2}{6},\quad \tilde\sigma_{22}^*=\frac 1{6}, \quad\tilde\sigma_{12}^*=0.\]
\end{theorem}

\begin{remark}\label{lem1}
When (\ref{cond1}) holds with $m(s)=\alpha$, we can show the rate of
convergence for $\hat\alpha^*$ is faster than the rate of
convergence for $\hat\alpha$. That is, the estimator based on
(\ref{rho}) is preferred to that based on (\ref{spearman}). However, the simulation study below  prefers the estimation procedure based on equation (\ref{spearman}) when the mean squared error is used as a criterion.
\end{remark}

\subsection{Nonparametric inference}

First we use (\ref{spearman}) to estimate the smooth function
$Q(s)=\frac 1{2\pi}\arcsin(\frac{1-m(s)/\log n}2)$
nonparametrically. Especially we consider the local linear estimator
$\hat Q(s)$ defined as
\[(\hat Q(s),\hat b)=\arg\min_{a,b}\sum_{i=1}^n \left((\hat F_1(X_i)-\frac 12)(\hat F_2(Y_i)-\frac12)-a-b(s-i/n)\right)k(\frac{s-i/n}h),\]
where $k$ is a kernel function and $h=h(n)\to 0$ is a bandwidth.
That is,
\[\hat Q(s)=\frac{\sum_{j=1}^nw_j(\hat F_1(X_j)-\frac 12)(\hat F_2(Y_j)-\frac 12)}{\sum_{j=1}^nw_j},\]
where $w_j=k(\frac{s-j/n}h)[s_{n,2}-(s-j/n)s_{n,1}]$,
$s_{n,l}=\sum_{j=1}^nk(\frac{s-j/n}h)(s-j/n)^l$. We refer to Fan and
Gijbels \cite{fan} for details. Therefore we can estimate $m(s)$ non parametrically by
\[\hat m(s)=\Big(1-2\sin(2\pi\hat Q(s))\Big)\log n.\]

\begin{theorem}\label{th4} Assume $k(s)$ is symmetric with support
$[-1, 1]$.  For a given $s\in (0,1 )$,  assume $m''(t)$ is
continuous at $s$, $h=h(n)\to 0$ and $\frac{h^2\sqrt{nh}}{\log n}\to
\lambda$ as $n\to\infty$. Then as $n\to\infty$ we have
\[\frac{\sqrt{nh}}{\log n}\Big(\hat m(s)-m(s)\Big)
\overset{d}{\to} N\left(\frac 12\lambda m^{\prime
\prime}(s)\int_{-1}^1t^2k(t)\,dt,~\frac{\pi^2}{15}\int_{-1}^1k^2(t)\,dt\right).\]
\end{theorem}

Second we use (\ref{rho}) to estimate the smooth function $m(s)$
nonparametrically by considering the local linear estimator
\[(\hat m^*(s),\hat b)=\arg\min_{a,b}\sum_{i=1}^n \left(\Phi^-(\hat F_1(X_i))\Phi^-(\hat F_2(Y_i))-1+\frac{a}{\log n}+\frac{b}{\log n}(s-i/n)\right)k(\frac{s-i/n}h),\]
i.e.,
\[\hat m^*(s)=-\frac{\sum_{j=1}^nw_j(\Phi^-(\hat F_1(X_j))\Phi^-(\hat F_2(Y_j))-1)\log n}{\sum_{j=1}^nw_j}.\]

\begin{theorem}\label{th4a} Assume $k(s)$ is symmetric with support
$[-1, 1]$.  For a given $s\in (0,1 )$,  assume $m''(t)$ is
continuous at $s$, $h=h(n)\to 0$ and $\frac{h^2\sqrt{nh}}{\log n}\to
\lambda$ as $n\to\infty$. Then as $n\to\infty$ we have
\[\frac{\sqrt{nh}}{\log n}(\hat m^*(s)-m(s))\overset{d}{\to} N\left(\frac 12\lambda m''(s)\int_{-1}^1t^2k(t)\,dt,~2\int_{-1}^1k^2(t)\,dt\right).\]
\end{theorem}

\begin{remark}\label{rem2}
It follows from Theorems \ref{th4} and \ref{th4a} that both $\hat
m^*(s)$ and $\hat m(s)$ have the same asymptotic bias, but $\hat m(s)$ has a
smaller asymptotic variance than $\hat m^*(s)$. Hence, unlike parametric
estimation, nonparametric estimation based on (\ref{spearman}) is always
preferred.
\end{remark}

\begin{remark}\label{rem3}
By minimizing the asymptotic mean squared error,  the optimal
choices of $h$ for $\hat m(s)$ and $\hat m^*(s)$ are
\[h_0=\left(\frac{\log^2n}n\right)^{1/5}\left(\frac{\pi^2\int_{-1}^1k^2(t)\,dt}{15(m''(s)\int_{-1}^1t^2k(t)\,dt)^2}\right)^{1/5}\]
and\[
h_0^*=\left(\frac{\log^2n}n\right)^{1/5}\left(\frac{2\int_{-1}^1k^2(t)\,dt}{(m''(s)\int_{-1}^1t^2k(t)\,dt)^2}\right)^{1/5},
\]
respectively, which are different from the standard optimal  order
$n^{-1/5}$  in the bandwidth choice of nonparametric regression
estimation and nonparametric density estimation. Data driven method
for choosing the above $h_0$ and $h_0^*$ can be obtained via
estimating $m''(s)$. A future research is to  investigate the
possibility of using cross-validation method to choose the optimal
bandwidth.
\end{remark}

\begin{remark}\label{rem4}
It is straightforward to construct both parametric and nonparametric
estimation for the tail dependence function and the tail coefficient
given in Theorem \ref{th1} and to derive the corresponding
asymptotic results by using Theorems \ref{th2}-\ref{th4a}.
\end{remark}

\section{Simulation}
In this section we examine the finite sample performance of the
proposed estimators by drawing independent $(X_1,Y_1), \cdots,
(X_n,Y_n)$ with $(X_i,Y_i)$ following the normal copula with
correlation coefficient  $\rho=1-m(i/n)/\log n$. We consider $n=300$ or $1000$ or $3000$,
and repeat $1000$ times.

First we consider $m(s)=\alpha$ with $\alpha=1$ or $10$, and
calculate the average, sample variance and mean squared error for both $\hat\alpha$ and
$\hat\alpha^*$.
Table 1 below  shows that $\hat\alpha^*$ has a smaller variance than
$\hat\alpha$, which confirms the argument  mentioned in Remark 2 that estimator $\hat\alpha^*$
has a faster rate of convergence than $\hat\alpha$. We also observe from Table 1  that i) $\hat\alpha^*$ has a larger bias and a larger mean squared error than $\hat\alpha$ except the case of $\alpha=10$ and $n=3000$; ii) the  variance and mean squared error of both $\hat\alpha$
and $\hat\alpha^*$ become larger when $\alpha$ increases; iii)   the accuracy for both estimators improves as $n$ becomes larger. In conclusion, $\hat\alpha$ has an overall  better finite sample behavior in terms of mean squared error than $\hat\alpha^*$ although its asymptotic variance is larger theoretically and empirically.

\begin{table}[!h]
\label{T1}\caption{Estimators for the case of $m(s)=\alpha$.}
\vspace{0.1in}
\begin{center}
\begin{tabular}{ccccccc}
\hline
&$\alpha=1$&$\alpha=10$& $\alpha=1$ &$\alpha=10$ &$\alpha=1$ &$\alpha=10$ \\
&$n=300$&$n=300$&$n=1000$&$n=1000$&$n=3000$&$n=3000$\\
\hline
$\E(\hat\alpha)$&1.0365&9.9660&1.0145&9.9836&1.0065&9.9976\\
$\V(\hat\alpha)$&0.0144&0.0249&0.0045&0.0384&0.0016&0.0218\\
$\text{MSE}(\hat\alpha)$&0.0157&0.0203&0.0047&0.0387&0.0016&0.0218\\
\hline
$\E(\hat\alpha^*)$&1.1690&9.8440&1.0788&9.9476&1.0368&9.9914\\
$\V(\hat\alpha^*)$&0.0109&0.0191&0.0034&0.0322&0.0012&0.0193\\
$\text{MSE}(\hat\alpha^*)$&0.0395&0.0434&0.0096&0.0349&0.0026&0.0194\\
\hline
\end{tabular}
\end{center}
\end{table}

Next we consider the case of $m(s)=\alpha+\beta s$. In Table 2 we
report the average, sample variance and mean squared error for estimators $(\hat\alpha,
\hat\beta)$, $(\hat\alpha^*, \hat\beta^*)$,
$(\hat\alpha+\frac{\hat\beta}2,
\frac{\hat\alpha}2+\frac{\hat\beta}3)$ and
$(\hat\alpha^*+\frac{\hat\beta^*}2,
\frac{\hat\alpha^*}2+\frac{\hat\beta^*}3)$. As we see, estimators
$(\hat\alpha, \hat\beta)$ have a smaller variance than
$(\hat\alpha^*, \hat\beta^*)$, but
$\hat\alpha^*+\frac{\hat\beta^*}2$ has a smaller variance than
$\hat\alpha+\frac{\hat\beta}2$, which is supported by Theorems
\ref{th3} and \ref{th3a} that $\hat\alpha^*+\frac{\hat\beta^*}2$ has
a faster rate of convergence than $\hat\alpha+\frac{\hat\beta}2$.  As
$n$ becomes larger, the accuracy of all estimators improves. Since  $\hat\alpha$
and $\hat\beta$ have a smaller mean squared error than $\hat\alpha^*$ and $\hat\beta^*$, respectively, we prefer the estimation procedure based on equation (\ref{spearman}) to that based on equation (\ref{rho}).

\begin{table}
\label{T2}\caption{Estimators for the case of $m(s)=\alpha+\beta s$
with $\alpha=1$.} \vspace{0.1in}
\begin{center}
\begin{tabular}{ccccccc}
\hline
& $\beta=1$ &$\beta=0$& $\beta=1$ &$\beta=0$ &$\beta=1$ &$\beta=0$ \\
&$n=300$&$n=300$&$n=1000$&$n=1000$&$n=3000$&$n=3000$\\
\hline
$\E(\hat\alpha)$&1.0289&1.0270&1.0350&1.0266&1.0019&1.0000\\
$\V(\hat\alpha)$&0.2901&0.3230&0.1245&0.1327&0.0453&0.0528\\
$\text{MSE}(\hat\alpha)$&0.2909&0.3237&0.1257&0.1334&0.0453&0.0528\\
$\E(\hat\beta)$&1.0345&0.0486&0.9612&-0.0240&1.0022&0.0111\\
$\V(\hat\beta)$&1.1597&1.2678&0.5080&0.5095&0.1833&0.2082\\
$\text{MSE}(\hat\beta)$&1.1609&1.2702&0.5095&0.5101&0.1833&0.2083\\
\hline
$\E(\hat\alpha+\frac{\hat\beta}2)$&1.5461&1.0513&1.5030&1.0146&1.5030&1.0055\\
$\V(\hat\alpha+\frac{\hat\beta}2)$&0.0270&0.0150&0.0097&0.0047&0.0032&0.0015\\
$\text{MSE}(\hat\alpha+\frac{\hat\beta}2)$&0.0291&0.0176&0.0097&0.0049&0.0041&0.0015\\
$\E(\frac{\hat\alpha}2+\frac{\hat\beta}3)$&0.8593&0.5297&0.8379&0.5053&0.8350&0.5037\\
$\V(\frac{\hat\alpha}2+\frac{\hat\beta}3)$&0.0170&0.0131&0.0070&0.0047&0.0024&0.0019\\
$\text{MSE}(\frac{\hat\alpha}2+\frac{\hat\beta}3)$&0.0177&0.0140&0.0070&0.0047&0.0024&0.0019\\
\hline
$\E(\hat\alpha^*)$&1.1654&1.1557&1.1155&1.0880&1.0349&1.0292\\
$\V(\hat\alpha^*)$&0.4281&0.4482&0.2303&0.2246&0.0934&0.1121\\
$\text{MSE}(\hat\alpha^*)$&0.4555&0.4724&0.2436&0.2323&0.0946&0.1130\\
$\E(\hat\beta^*)$&0.9802&0.0338&0.9188&-0.0177&0.9931&0.0147\\
$\V(\hat\beta^*)$&1.6477&1.7563&0.9138&0.8838&0.3686&0.4504\\
$\text{MSE}(\hat\beta^*)$&1.6481&1.7575&0.9204&0.8841&0.3686&0.4506\\
\hline
$\E(\hat\alpha^*+\frac{\hat\beta^*}2)$&1.6555&1.1726&1.5749&1.0792&1.5315&1.0365\\
$\V(\hat\alpha^*+\frac{\hat\beta^*}2)$&0.0221&0.0112&0.0076&0.0037&0.0025&0.0012\\
$\text{MSE}(\hat\alpha^*+\frac{\hat\beta^*}2)$&0.0463&0.0410&0.0132&0.0100&0.0035&0.0025\\
$\E(\frac{\hat\alpha^*}2+\frac{\hat\beta^*}3)$&0.9094&0.5891&0.8640&0.5381&0.8485&0.5195\\
$\V(\frac{\hat\alpha^*}2+\frac{\hat\beta^*}3)$&0.0174&0.0152&0.0087&0.0071&0.0033&0.0036\\
$\text{MSE}(\frac{\hat\alpha^*}2+\frac{\hat\beta^*}3)$&0.0232&0.0231&0.0096&0.0086&0.0035&0.0040\\
\hline
\end{tabular}
\end{center}
\end{table}

Finally we consider the case of $m(s)=\alpha+\beta s^{\gamma}$. Given results in Tables 1 and 2, we only consider the estimators derived from equation (\ref{spearman}) with the large sample size $n=3000$. Table 3 shows that all estimators  have a rather large variance for $\gamma=1$, and the variance of $\hat\gamma$ is still quite big even when $\gamma=0.5$, which means estimating the shape parameter $\gamma$ is very challenging as usually.

\begin{table}[!h]
\label{T3}\caption{Estimators for the case of $m(s)=\alpha+\beta s^{\gamma}$ with $\alpha=\beta=1$.}
\vspace{0.1in}
\begin{center}
\begin{tabular}{cccc|ccc|ccc}
\hline
&$\E(\hat\alpha)$&$\V(\hat\alpha)$& $\text{MSE}(\hat\alpha)$&$\E(\hat\beta)$&$\V(\hat\beta)$&$\text{MSE}(\hat\beta)$&$\E(\hat\gamma)$&$\V(\hat\gamma)$&$\text{MSE}(\hat\gamma)$ \\
\hline
$\gamma=0.5$&0.8631&0.2050&0.2237&1.2268&0.2840&0.3354&0.9787&8.0693&8.2985\\
$\gamma=1$&0.9964&15.8532&15.8532&1.1412&16.2379&16.2578&1.7859&11.1177&11.7353\\
\hline
\end{tabular}
\end{center}
\end{table}

\section{Data Analysis}
In this section we apply the proposed nonparametric estimators to
two real data sets: Danish fire loss and log-returns of exchange
rates; see Figure 1.

This first data set is the nonzero losses to building and content in
the Danish fire insurance claims, which comprises 2167 fire losses
over the period 1980 to 1990. The second data set is the log-returns
of the exchange rates between Euro and US dollar and those between
British pound and US dollar from January 3, 2000 till December 19,
2007.
\begin{figure}[!h]
\centering
\includegraphics[width=7cm,height=4cm]{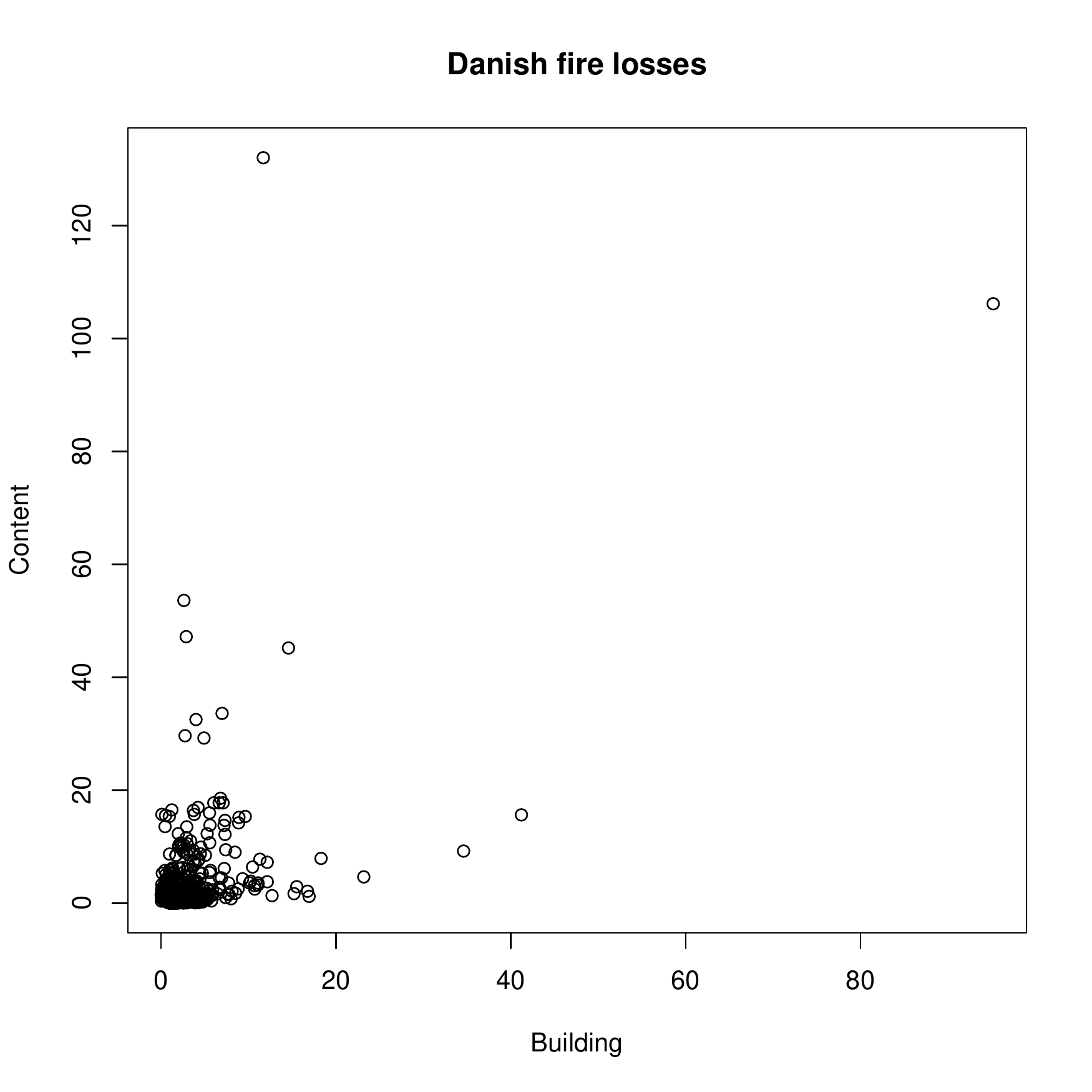}
\includegraphics[width=7cm,height=4cm]{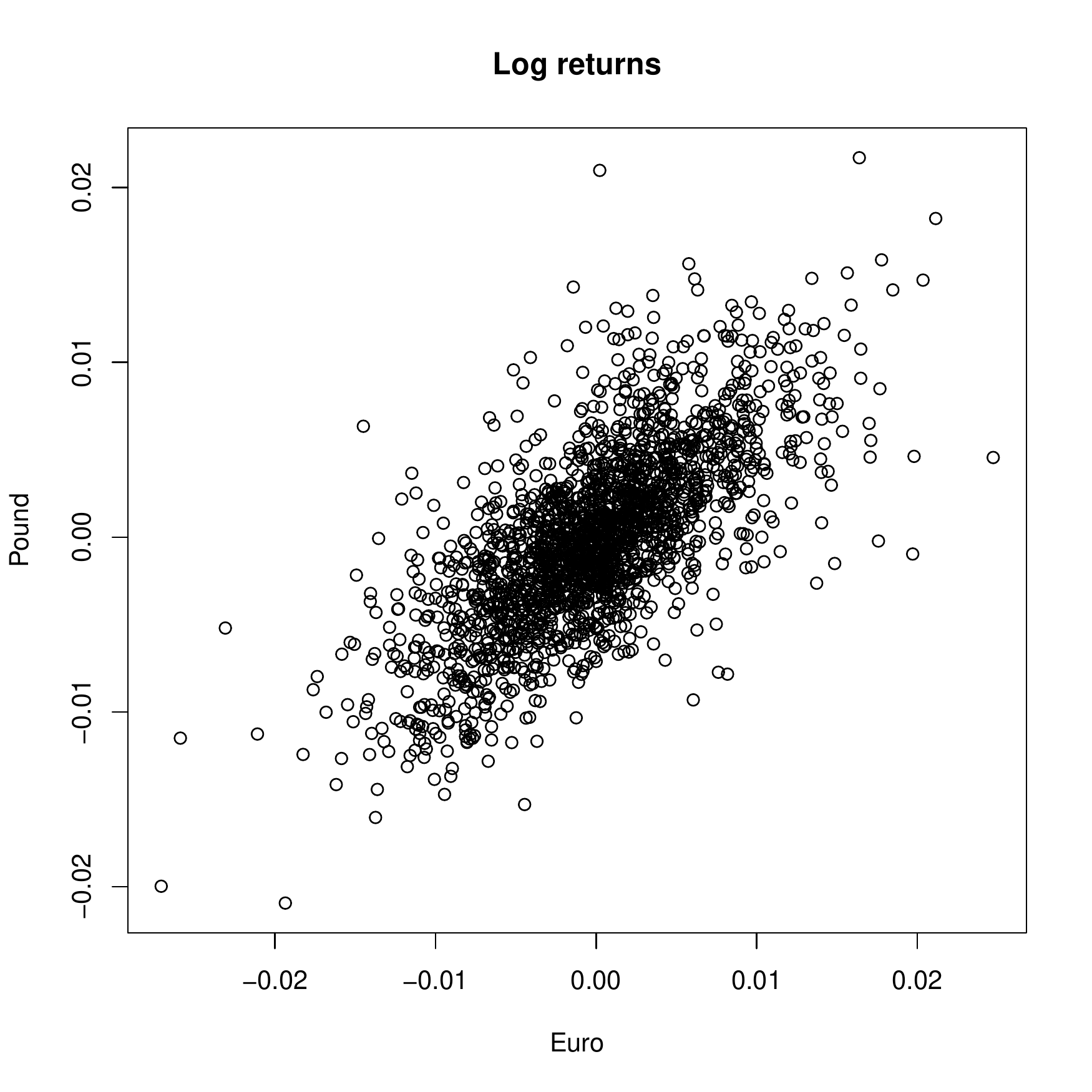}
\caption{Left panel: Danish fire loss with 2167 fire losses over the
period 1980 to 1990. Right panel: log-returns of exchange rates
between Euro and US dollar and those between British pound and US
dollar from January 3, 2000 till December 19, 2007.}
\end{figure}

We calculate both $\hat m(s)$ and $\hat m^*(s)$ for
$s=0.1,0.11,0.12,\cdots,0.9$ by using Epanechnikov kernel
$k(x)=\frac 34 (1-x^2)I(|x|\le 1)$ and the bandwidth
$h=d\{\log^2(n)/n\}^{1/5}$ with $d=0.2, 0.3, 0.4, 0.5$. From Figures
2 and 3, we observe that $\hat m(s)$ and $\hat m^*(s)$ have a quite
similar pattern for the second data set, but seem having a different
pattern for the first data set when a large bandwidth is employed. To further investigate this issue, we plot the difference of $\hat m(s)-\hat m^*(s)$ in Figure 4 for the above $h$ with $d=0.2, 0.3, 0.4, 0.5$, which indeed shows the differences for $d=0.4$ are quite similar to those for $d=0.5$.  Nevertheless, Remark 3 says that one should prefer $\hat m(s)$ to $\hat m^*(s)$. The non-constant $m(s)$ function indicates observations are not identically distributed.

\begin{figure}[!h]
\centering
\includegraphics[width=7cm,height=4cm]{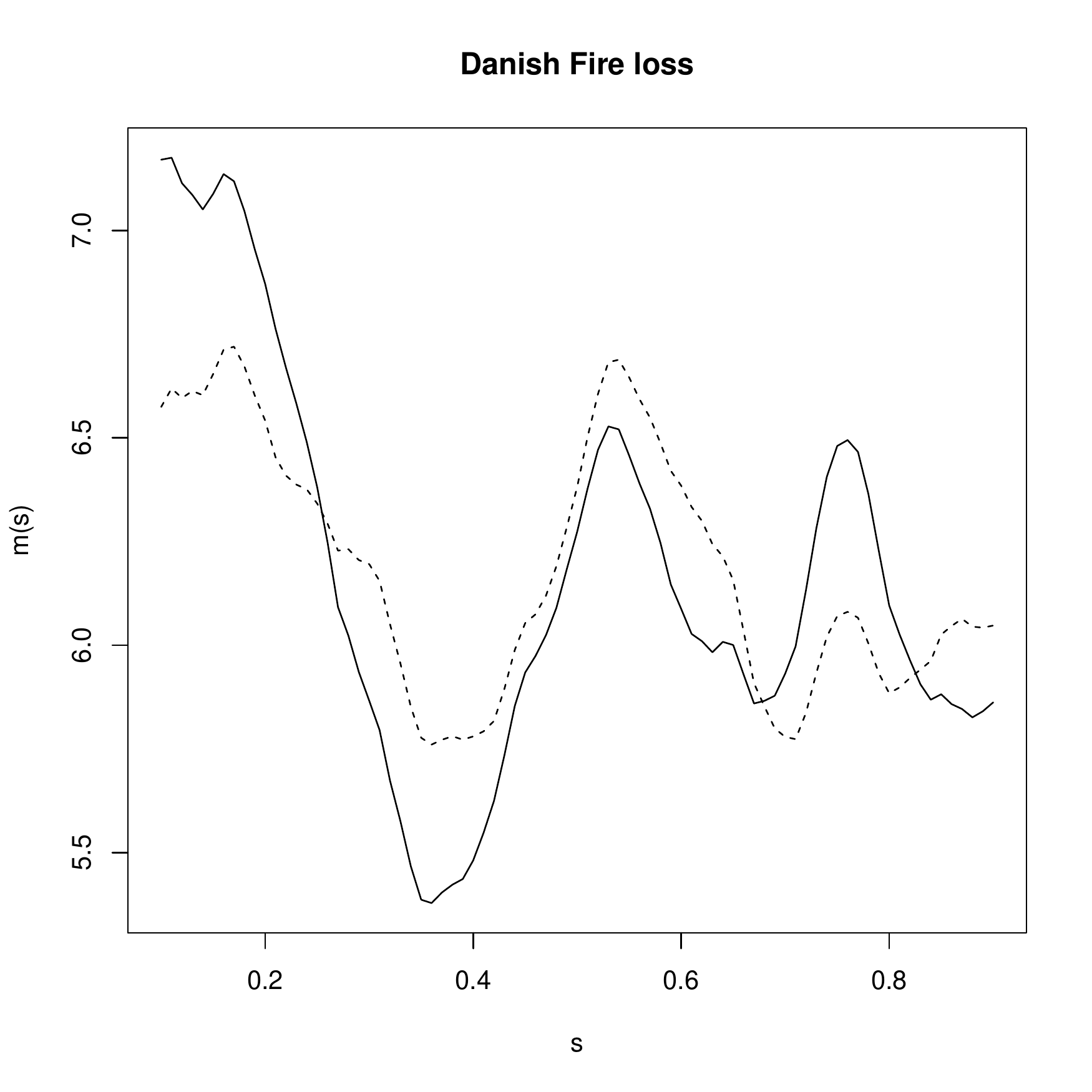}
\includegraphics[width=7cm,height=4cm]{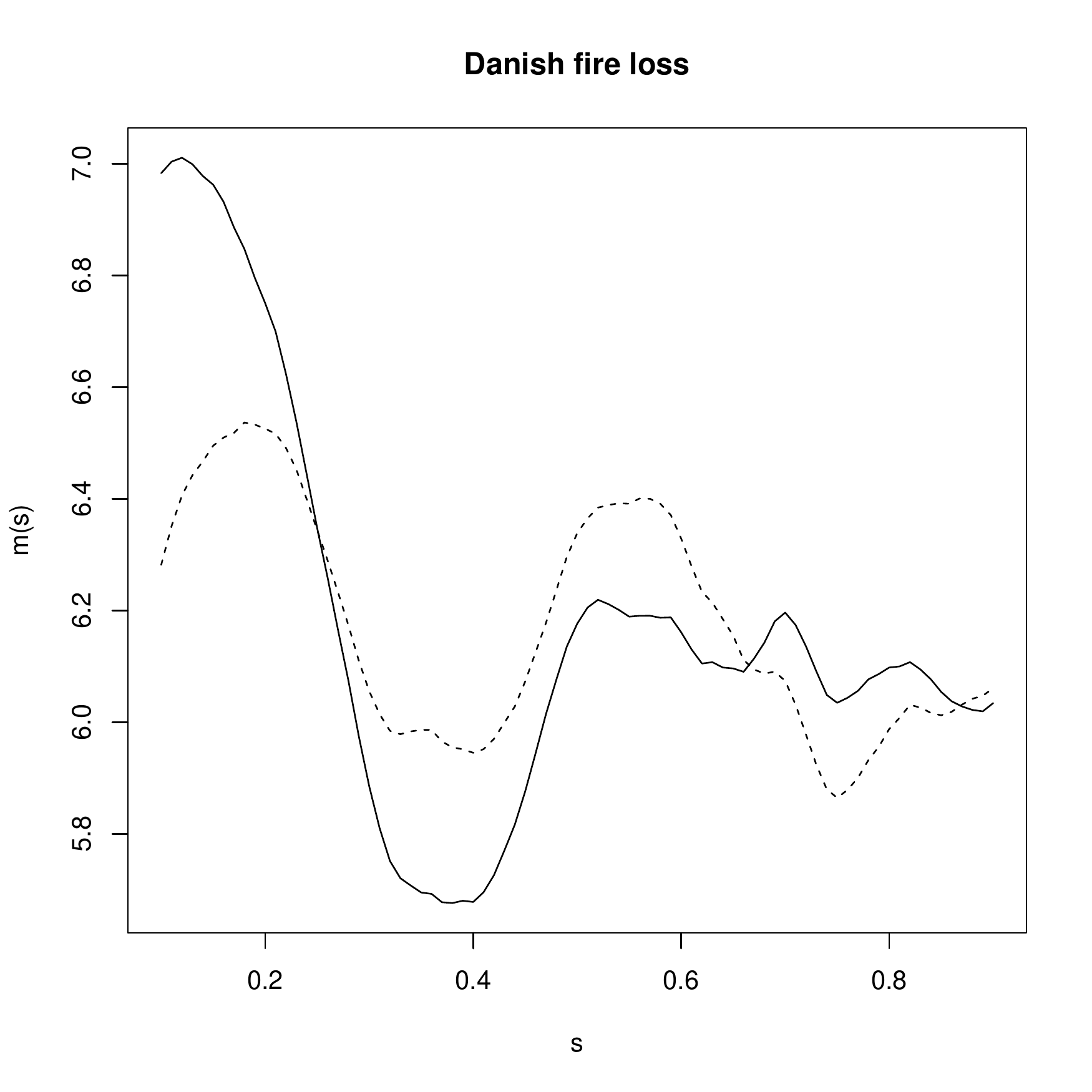}
\includegraphics[width=7cm,height=4cm]{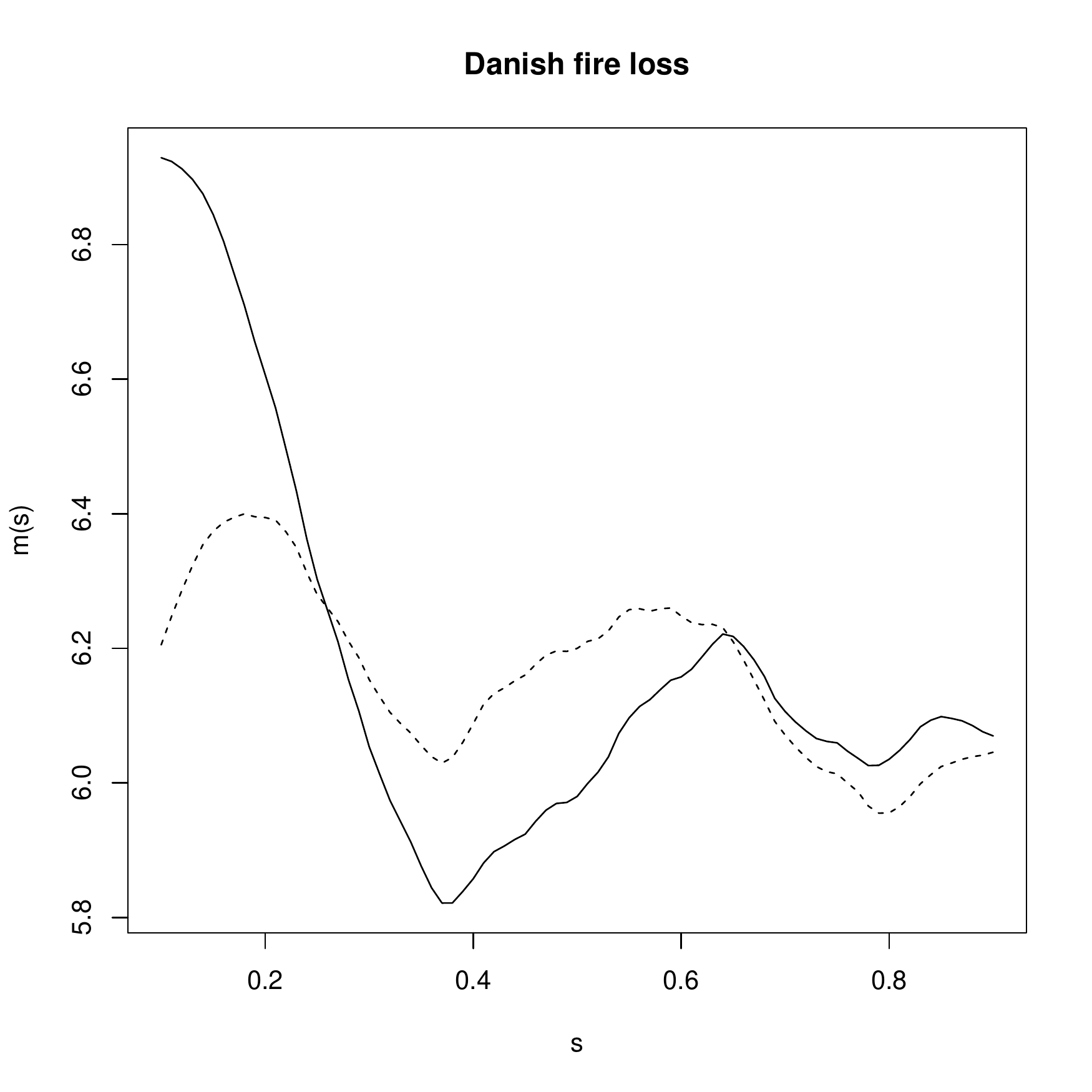}
\includegraphics[width=7cm,height=4cm]{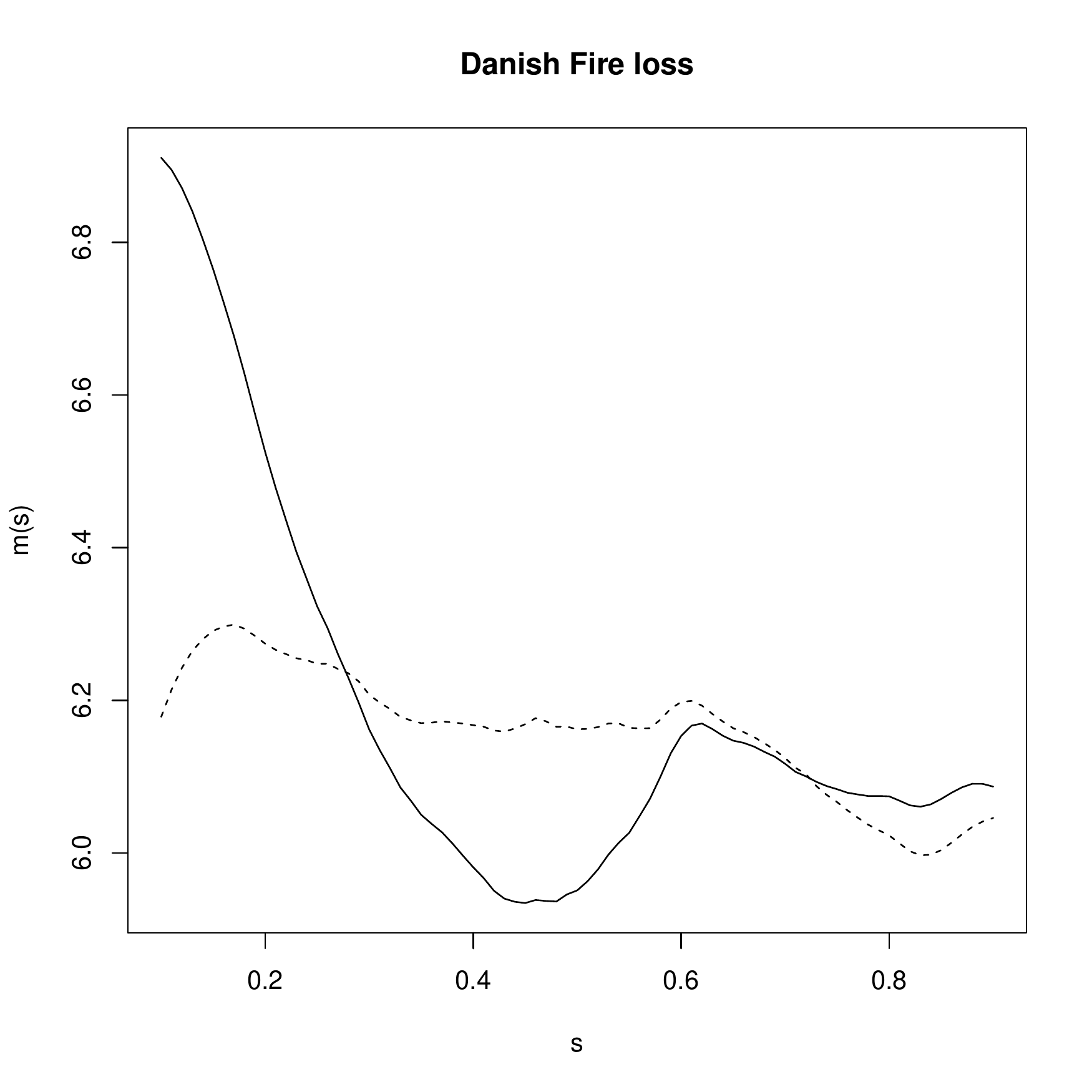}
\caption{Danish fire losses. Solid line and dotted line represent $\hat m(s)$ and $\hat
m^*(s)$, respectively. Bandwidth $h=d \{\log^2(n)/n\}^{1/5}$ with
$d=0.2, 0.3, 0.4, 0.5$ is employed in the upper left, upper right,
lower left, lower right panels, respectively.}
\end{figure}

\begin{figure}[!h]
\centering
\includegraphics[width=7cm,height=4cm]{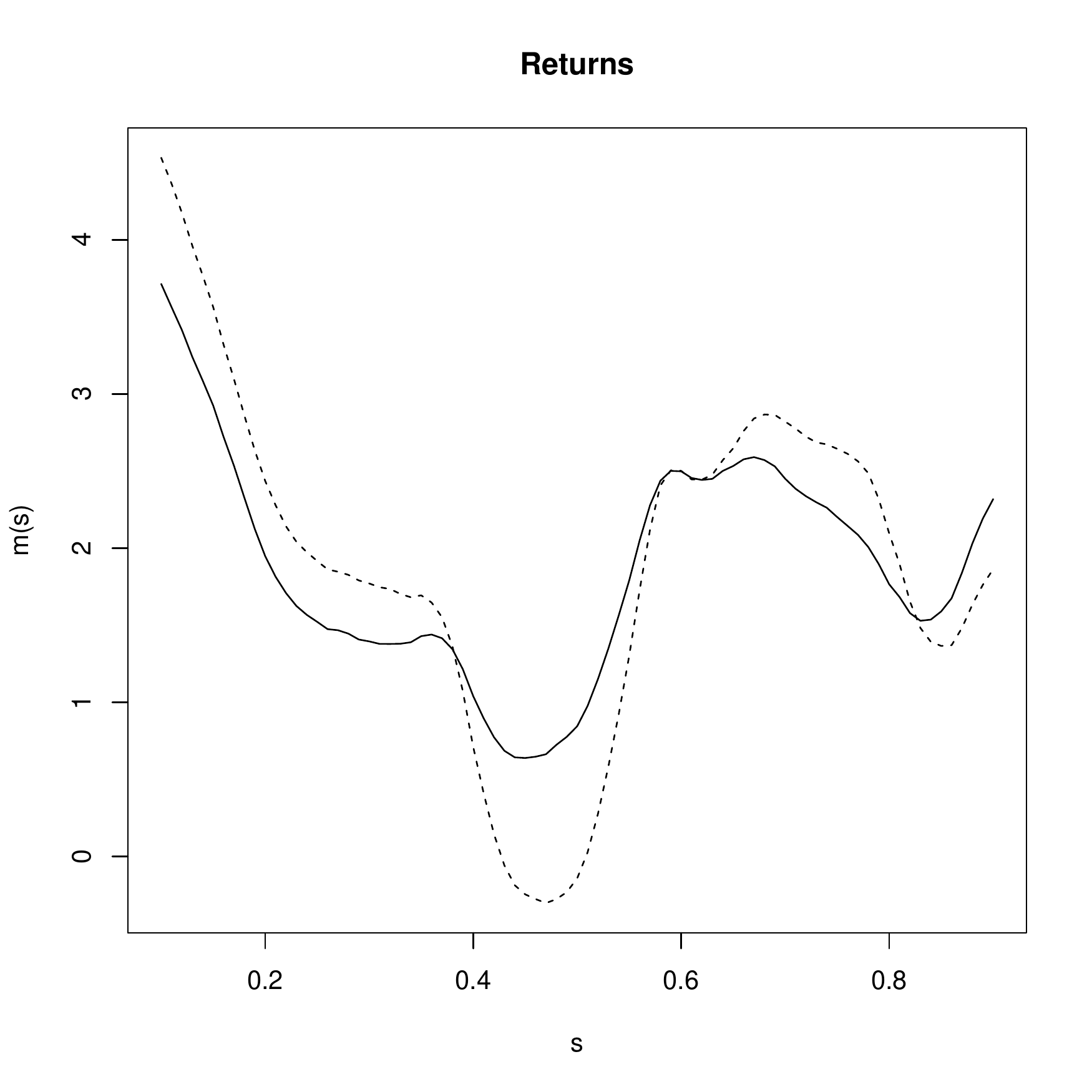}
\includegraphics[width=7cm,height=4cm]{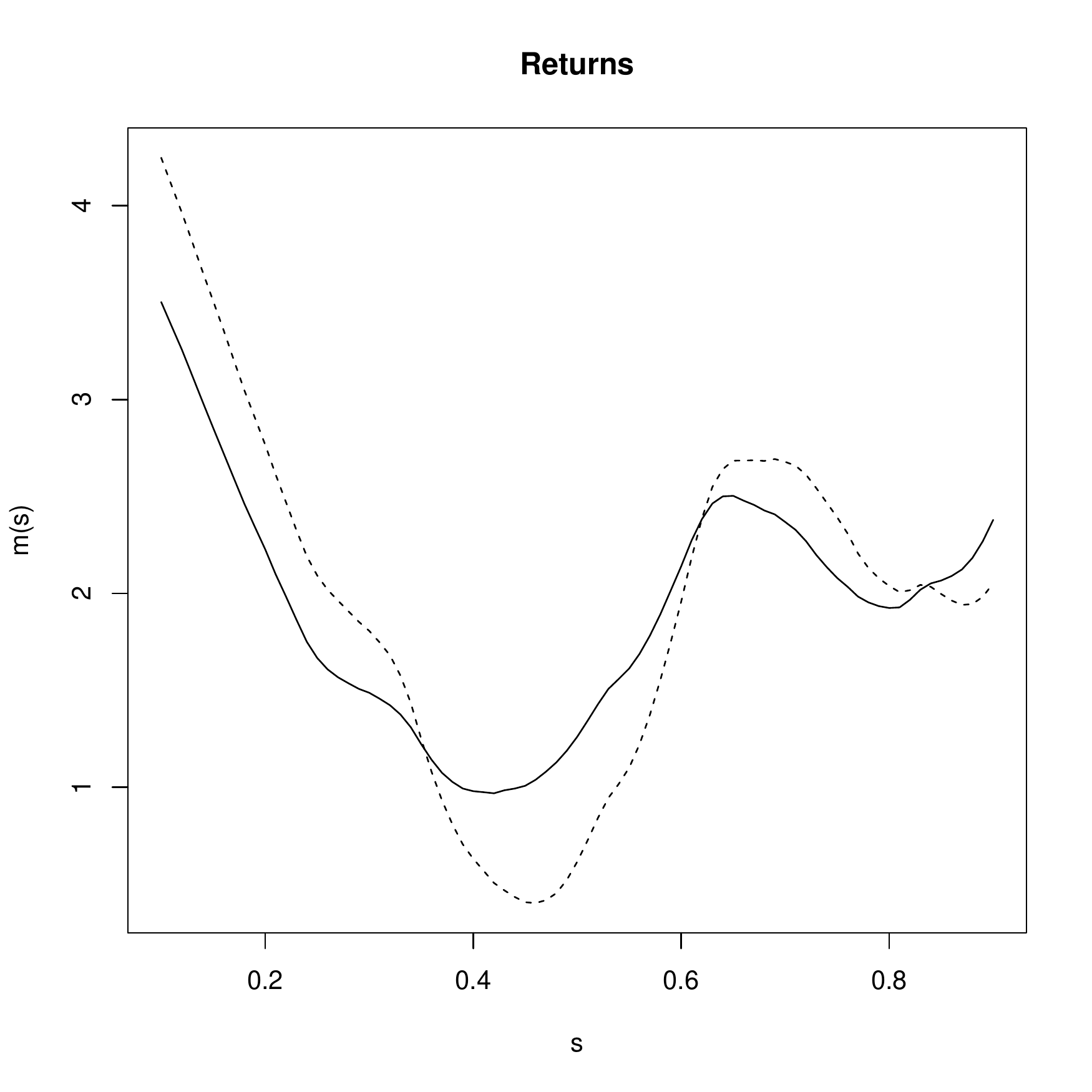}
\includegraphics[width=7cm,height=4cm]{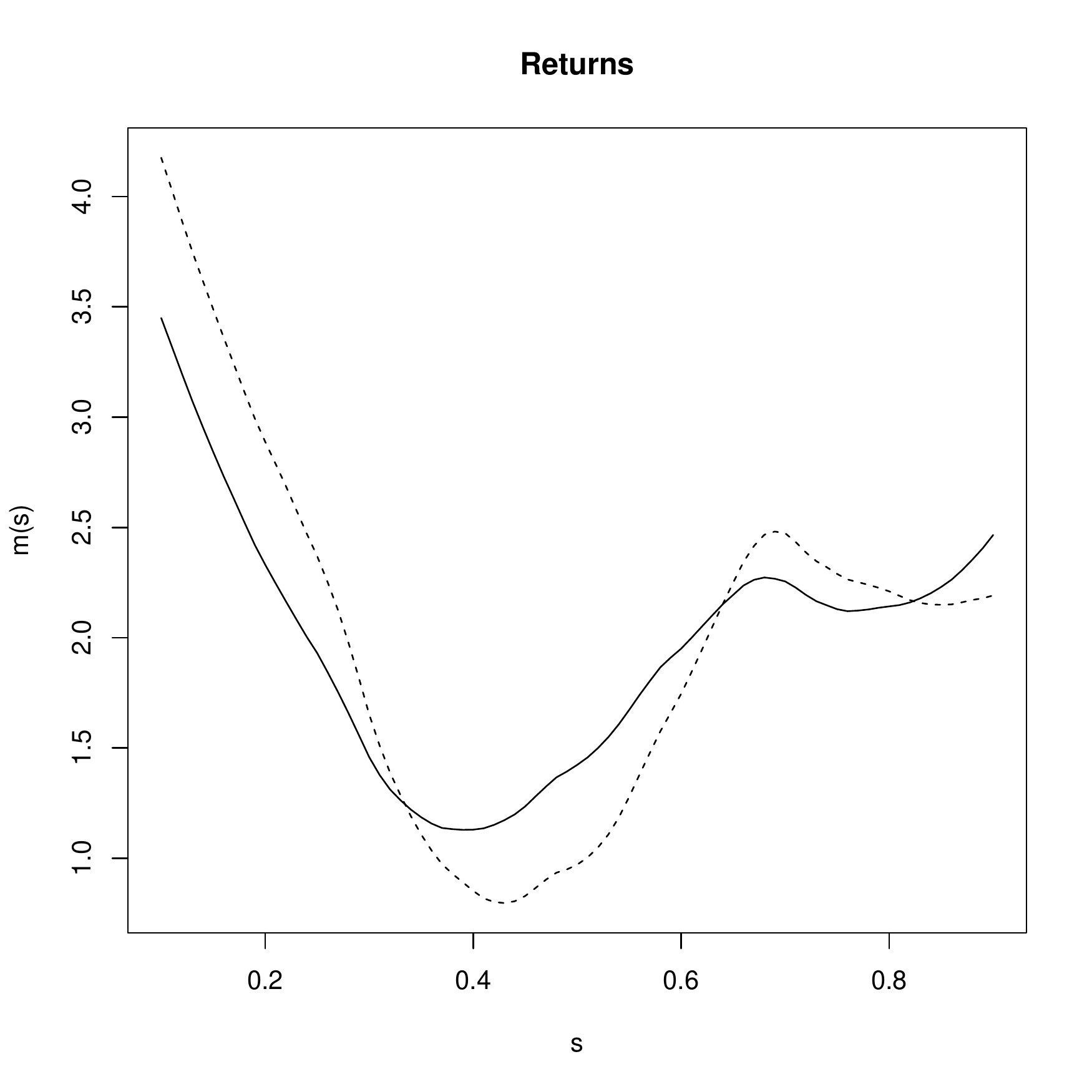}
\includegraphics[width=7cm,height=4cm]{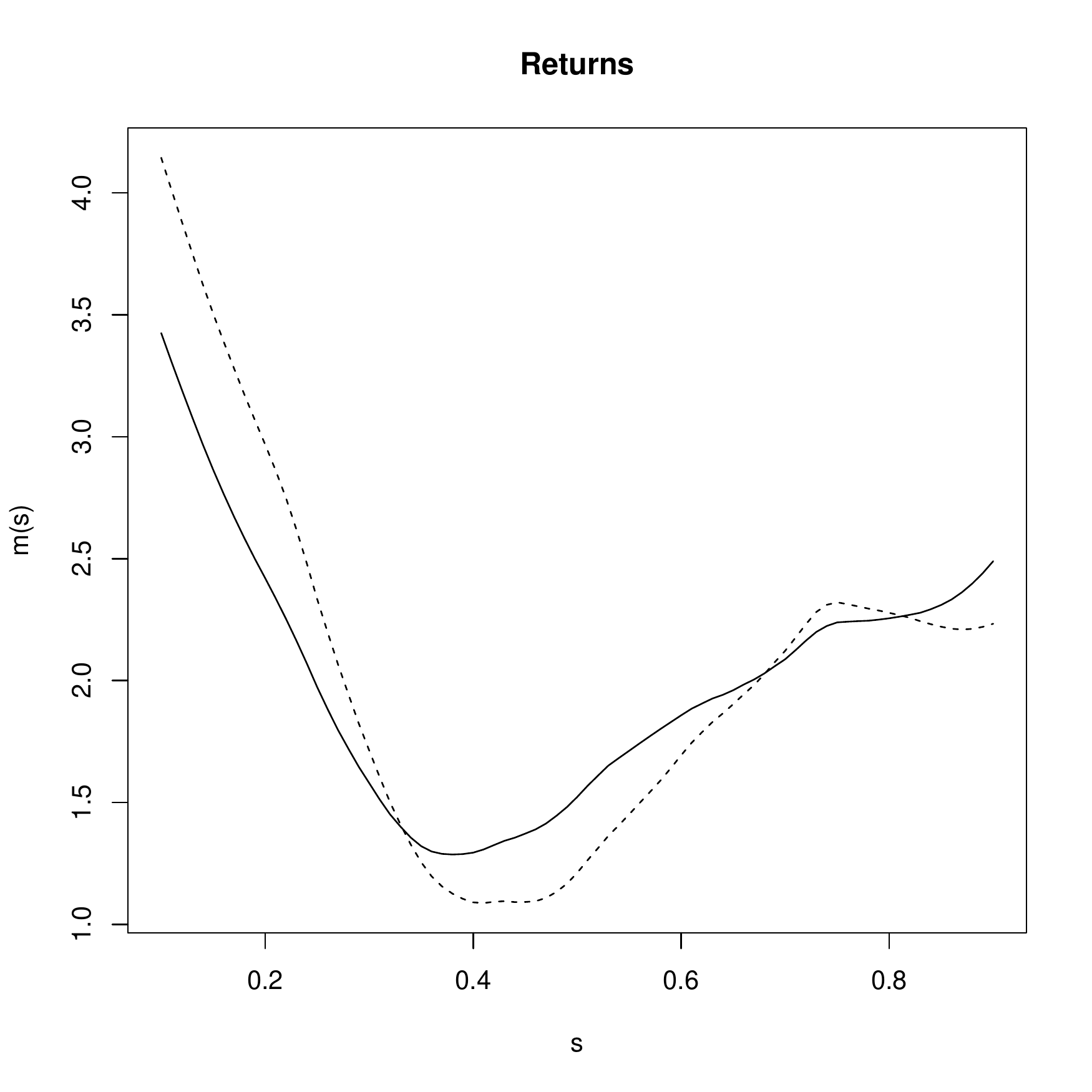}
\caption{Exchange rates. Solid line and dotted line represent $\hat m(s)$ and $\hat
m^*(s)$, respectively. Bandwidth $h=d \{\log^2(n)/n\}^{1/5}$ with
$d=0.2, 0.3, 0.4, 0.5$ is employed in the upper left, upper right,
lower left, lower right panels, respectively.}
\end{figure}

\begin{figure}[!t]
\centering
\includegraphics[width=14cm,height=6cm]{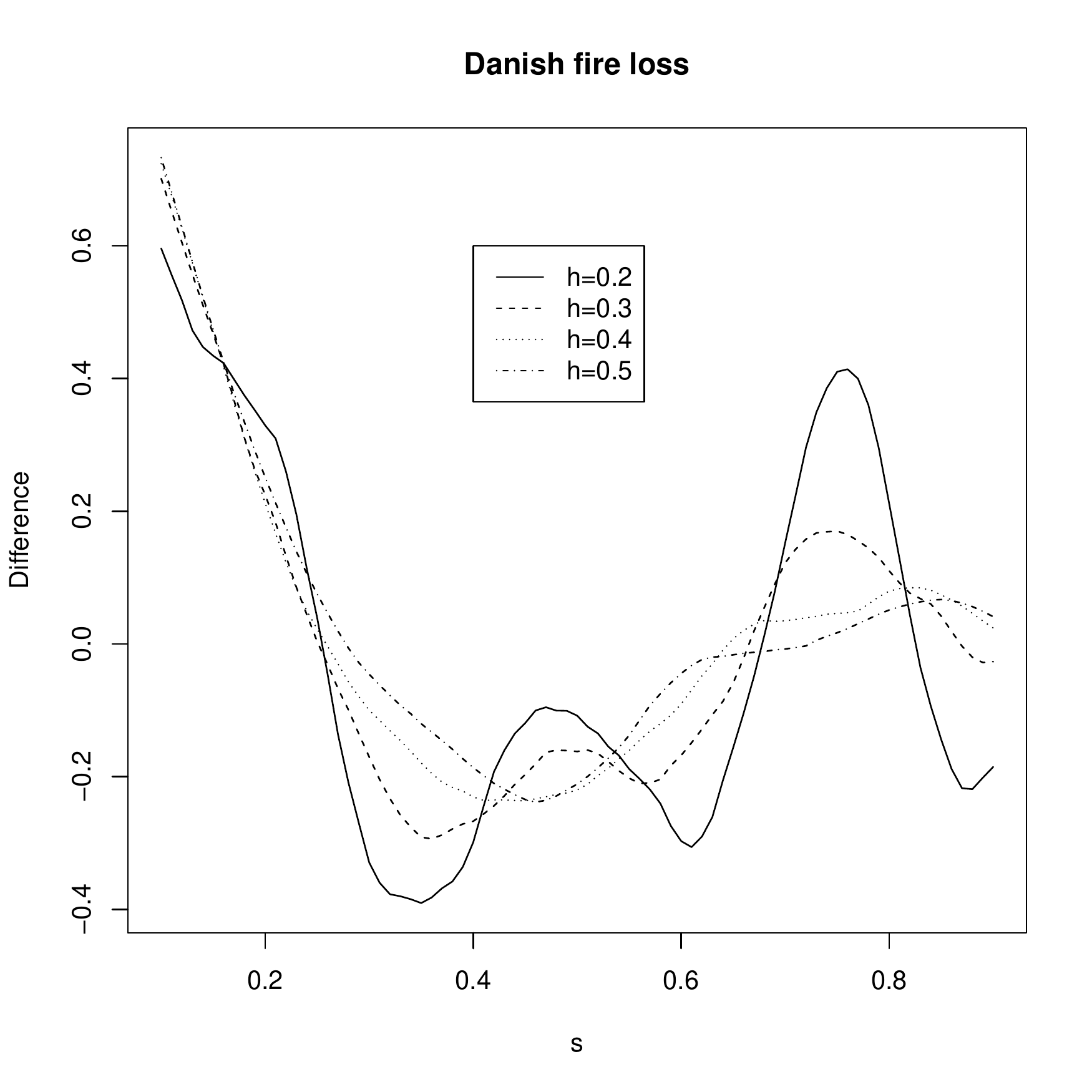}

\caption{Danish fire losses. Differences of $\hat m(s)-\hat
m^*(s)$ are plotted for bandwidth $h=d \{\log^2(n)/n\}^{1/5}$ with
$d=0.2, 0.3, 0.4, 0.5$.}
\end{figure}

\section{Proofs}

\begin{proof}[Proof of Theorem 1.]

We focus on the proof of case iii) since the other two cases can be
verified easily.

For any $\epsilon>0$ such that $\epsilon<-x$, write
\begin{equation}\label{pfTh1-1}\begin{array}{ll}
&1-\P(\Phi^-(F_1(X_i))\le \Phi^-(1+\frac xn), \Phi^-(F_2(Y_i))\le \Phi^-(1+\frac yn))\\
=&-\frac xn-\frac yn-\P\left(\Phi^{-}(F_1(X_i))>\Phi^-(1+\frac xn), \Phi^{-}(F_2(Y_i))>\Phi^-(1+\frac yn)\right)\\
=&-\frac xn-\frac yn-\int_{\Phi^-(1+x/n)}^{\infty}\left(1-\Phi\Big(\frac{\Phi^-(1+y/n)-\rho_is}{\sqrt{1-\rho_i^2}}\Big)\right)\,d\Phi(s)\\
=&-\frac xn-\frac yn-n^{-1}\int_x^0\left(1-\Phi\Big(\frac{\Phi^-(1+y/n)-\rho_i\Phi^-(1+s/n)}{\sqrt{1-\rho_i^2}}\Big)\right)\,ds\\
=&-\frac yn
-n^{-1}\int_{-\epsilon}^x\Phi\Big(\frac{\Phi^-(1+y/n)-\rho_i\Phi^-(1+s/n)}{\sqrt{1-\rho_i^2}}\Big)\,ds
-n^{-1}\int_0^{-\epsilon}\Phi\Big(\frac{\Phi^-(1+y/n)-\rho_i\Phi^-(1+s/n)}{\sqrt{1-\rho_i^2}}\Big)\,ds.
\end{array}
\end{equation}

For fixed $x<0$ and $y<0$, we have
\[\Phi^-(1+y/n)=\sqrt{2\log n}-\frac{\log(-y)}{\sqrt{2\log n}}-\frac{\log\log n+\log(4\pi)}{2\sqrt{2\log n}}+o(\frac{\log\log n}{\sqrt{\log n}})\]
and
\[\Phi^-(1+s/n)=\sqrt{2\log n}-\frac{\log(-s)}{\sqrt{2\log n}}-\frac{\log\log n+\log (4\pi)}{2\sqrt{2\log n}}+o(\frac{\log\log n}{\sqrt{\log n}})\]
uniformly in $s\in [x, -\epsilon]$, which implies that
\begin{equation}\label{pfTh1-2}\begin{array}{ll}
&\frac{\Phi^-(1+y/n)-\rho_i\Phi^-(1+s/n)}{\sqrt{1-\rho_i^2}}\\
=&\frac{\sqrt{2\log
n}\sqrt{1-\rho_i}}{\sqrt{1+\rho_i}}-\frac{\log(-y)}{\sqrt{2\log
n}\sqrt{1-\rho_i}\sqrt{1+\rho_i}}
+\frac{\rho_i\log(-s)}{\sqrt{2\log n}\sqrt{1-\rho_i}\sqrt{1+\rho_i}}-\frac{\log\log n+\log(4\pi)}{2\sqrt{2\log n}}\frac{\sqrt{1-\rho_i}}{\sqrt{1+\rho_i}}+o(\frac{\log\log n}{\sqrt{\log n}})\\
 =&\frac{\sqrt{2m(i/n)}}{\sqrt{2-m(i/n)/\log
n}}-\frac{\log(-y)}{\sqrt{2m(i/n)}\sqrt{2-m(i/n)/\log n}}+
\frac{(1-m(i/n)/\log n)\log (-s)}{\sqrt {2m(i/n)}\sqrt{2-m(i/n)/\log
n}}+o(\frac{\log\log n}{\sqrt{\log n}})
\end{array}\end{equation}
uniformly for $s\in [x,-\epsilon]$, where  $x<0$ and $y<0$ are fixed
and $\epsilon\in (0, -x)$ is  any given constant.

Since $m(s)$ is a continuous positive function, it follows from
(\ref{pfTh1-2}) that
\begin{equation}\label{pfTh1-3}
\begin{array}{ll}
&n^{-1}\int_{-\epsilon}^x\Phi\Big(\frac{\Phi^-(1+y/n)-\rho_i\Phi^-(1+s/n)}{\sqrt{1-\rho_i^2}}\Big)\,ds\\
=&n^{-1}\left(\int_{-\epsilon}^x\Phi\Big(\sqrt{m(i/n)}-\frac{\log(-y)}{2\sqrt{m(i/n)}}+\frac{\log(-s)}{2\sqrt{m(i/n)}}\Big)\,ds\right)(1+o(1))\\
=&\left(n^{-1}x\Phi\Big(\sqrt{m(i/n)}+\frac{\log(x/y)}{2\sqrt{m(i/n)}}\Big)+n^{-1}\epsilon\Phi\Big(\sqrt{m(i/n)}+\frac{\log(-\epsilon/y)}{2\sqrt{m(i/n)}}\Big)\right.\\
&\left. -n^{-1}\frac{1}{2\sqrt{m(i/n)}}\int_{-\epsilon}^x\phi\Big(\sqrt{m(i/n)}-\frac{\log(-y)}{2\sqrt{m(i/n)}}+\frac{\log(-s)}{2\sqrt{m(i/n)}}\Big)\,ds\right)(1+o(1))\\
=&\left(n^{-1}x\Phi\Big(\sqrt{m(i/n)}+\frac{\log(x/y)}{2\sqrt{m(i/n)}}\Big)+n^{-1}\epsilon\Phi\Big(\sqrt{m(i/n)}+\frac{\log(-\epsilon/y)}{2\sqrt{m(i/n)}}\Big)\right.\\
&\left. +n^{-1}\frac{1}{2\sqrt{m(i/n)}}\int_{\log\epsilon}^{\log(-x)}\phi\Big(\sqrt{m(i/n)}-\frac{\log(-y)}{2\sqrt{m(i/n)}}+\frac{s}{2\sqrt{m(i/n)}}\Big)e^s\,ds\right)(1+o(1))\\
=&\left(n^{-1}x\Phi\Big(\sqrt{m(i/n)}+\frac{\log(x/y)}{2\sqrt{m(i/n)}}\Big)+n^{-1}\epsilon\Phi\Big(\sqrt{m(i/n)}+\frac{\log(-\epsilon/y)}{2\sqrt{m(i/n)}}\Big)\right.\\
&\left. -n^{-1}y\Phi\Big(\frac{\log(x/y)}{2\sqrt{m(i/n)}}-\sqrt{m(i/n)}\Big)+n^{-1}y\Phi\Big(\frac{\log(-\epsilon/y)}{2\sqrt{m(i/n)}}-\sqrt{m(i/n)}\Big)\right)(1+o(1))\\
=&\left(n^{-1}x\Phi\Big(\sqrt{m(i/n)}+\frac{\log(x/y)}{2\sqrt{m(i/n)}}\Big)+n^{-1}\epsilon\Phi\Big(\sqrt{m(i/n)}+\frac{\log(-\epsilon/y)}{2\sqrt{m(i/n)}}\Big)\right.\\
&\left.-n^{-1}y+n^{-1}y\Phi\Big(\frac{\log(y/x)}{2\sqrt{m(i/n)}}+\sqrt{m(i/n)}\Big)+n^{-1}y\Phi\Big(\frac{\log(-\epsilon/y)}{2\sqrt{m(i/n)}}-\sqrt{m(i/n)}\Big)\right)(1+o(1)),
\end{array}\end{equation}
where $\phi(s)=\Phi'(s)$. Hence, it follows from  (\ref{pfTh1-1})
and (\ref{pfTh1-3}) that
\[\begin{array}{ll}
&\lim_{\epsilon\to 0}\lim_{n\to\infty}\sum_{i=1}^n\Big(1-\P(\Phi^-(F_1(X_i))\le \Phi^-(1+\frac xn), \Phi^-(F_2(Y_i))\le\Phi^-(1+\frac yn))\Big)\\
=&-x\int_0^1\Phi\Big(\sqrt{m(s)}+\frac{\log(x/y)}{2\sqrt{m(s)}}\Big)\,ds-y\int_0^1\Phi\Big(\sqrt{m(s)}+\frac{\log(y/x)}{2\sqrt{m(s)}}\Big)\,ds
\end{array}\]
for any $x<0$ and $y<0$, which implies that
\[\begin{array}{ll}
&\lim_{n\to\infty}\P\Big(n(\max_{1\le i\le n}F_1(X_i)-1)\le x, n(\max_{1\le i\le n}F_2(Y_i)-1)\le y\Big)\\
=&\lim_{n\to\infty}\prod_{i=1}^n\P\Big(\Phi^-(F_1(X_i))\le\Phi^-(1+\frac xn), \Phi^-(F_2(Y_i))\le \Phi^-(1+\frac yn)\Big)\\
=&\exp\left(\lim_{n\to\infty}\sum_{i=1}^n\log \P\Big(\Phi^-(F_1(X_i))\le\Phi^-(1+\frac xn), \Phi^-(F_2(Y_i))\le\Phi^-(1+\frac yn)\Big)\right)\\
=&\exp\left(-\lim_{n\to\infty}\sum_{i=1}^n\left(1-\P\Big(\Phi^-(F_1(X_i))\le\Phi^-(1+\frac xn), \Phi^-(F_2(Y_i))\le \Phi^-(1+\frac yn)\Big)\right)\right)\\
=&\exp\left(x\int_0^1\Phi\Big(\sqrt{m(s)}+\frac{\log(x/y)}{2\sqrt{m(s)}}\Big)\,ds+y\int_0^1\Phi\Big(\sqrt{m(s)}+\frac{\log
(y/x)}{2\sqrt{m(s)}}\Big)\,ds\right)
\end{array}\]
for all $x<0$ and $y<0$. The rest for computing the tail dependence
function and tail coefficient is straightforward.
\end{proof}

\begin{proof}[Proof of Theorem \ref{th2}.]

Put $U_i=F_1(X_i),$ $V_i=F_2(Y_i),$ $\hat U_n(u)=\frac
1{n+1}\sum_{i=1}^nI(U_i\le u)$, $\hat V_n(v)=\frac
1{n+1}\sum_{i=1}^nI(V_i\le v)$, $Z_i=\left(U_i-\frac
12\right)\left(V_i-\frac 12\right)$ and $\hat Z_i=\left(\hat
U_n(U_i)-\frac 12\right)\left(\hat V_n(V_i)-\frac 12\right)$ for
$i=1,\cdots,n$. Then
\begin{equation}\label{pfth2-1}
\left\{\left(\hat F_1(X_i)-\frac 12\right)\left(\hat F_2(Y_i)-\frac
12\right)\right\}_{i=1}^n\overset{d}{=}\{\hat Z_i\}_{i=1}^n.
\end{equation} It is also known that
\begin{equation}
\label{pfth2-2} \sup_{0<u<1}\left|\frac{\sqrt n\{\hat
U_n(u)-u\}}{u^\delta(1-u)^\delta}\right|=O_p(1)\quad\text{and}\quad\sup_{0<v<1}\left|\frac{\sqrt
n\{\hat V_n(v)-v\}}{v^\delta(1-v)^\delta}\right|=O_p(1);
\end{equation}
see Inequality 1 in Page 134 of Shorack and Wellner \cite{shorack}(1986).

Put
\[I_1=\frac 1{\sqrt n}\sum_{i=1}^n(\hat U_n(U_i)-U_i)(\hat V_n(V_i)-V_i),\]
\[\begin{array}{ll}
I_2&=\frac 1{(n+1)\sqrt n}\sum_{i=1}^n\sum_{j=1}^n\Big((I(U_j\le U_i)-U_i)(V_i-\frac 12)\\
&\quad-\int_0^1\int_0^1(I(U_j\le u)-u)(v-\frac
12)c(u,v;\rho_i)\,dudv\Big),\end{array}\]
\[\begin{array}{ll}
I_3&=\frac 1{(n+1)\sqrt n}\sum_{i=1}^n\sum_{j=1}^n\Big((I(V_j\le V_i)-V_i)(U_i-\frac 12)\\
&\quad-\int_0^1\int_0^1(I(V_j\le v)-v)(u-\frac
12)c(u,v;\rho_i)\,dudv\Big)\end{array}\] and
\[\begin{array}{ll}
\tilde Z_i&=\frac 1{(n+1)}\sum_{j=1}^n\int_0^1\int_0^1(I(U_i\le u)-u)(v-\frac 12)c(u,v;\rho_j)\,dudv\\
&\quad+\frac 1{(n+1)}\sum_{j=1}^n\int_0^1\int_0^1(I(V_i\le v)-v)(u-\frac 12)c(u,v;\rho_j)\,dudv\\
&\quad+\Big(Z_i-\frac 1{2\pi}\arcsin(\frac{\rho_i}2)\Big)\\
&=\tilde Z_{i,1}+\tilde Z_{i,2}+\tilde Z_{i,3}.
\end{array}\]
Therefore
\begin{equation}\label{pfth2-3}
\frac 1{\sqrt n}\sum_{i=1}^n\Big(\hat Z_i-\frac
1{2\pi}\arcsin(\frac{\rho_i}2)\Big) =I_1+I_2+I_3+\frac 1{\sqrt
n}\sum_{i=1}^n\tilde Z_i.
\end{equation}
It follows from (\ref{pfth2-2}) that
\begin{equation}\label{pfth2-4}
I_1=O_p(\frac 1{\sqrt n}).
\end{equation}
Direct calculations show that $\E I_2^2=O(\frac 1n)$ and $\E
I_3^2=O(\frac 1n)$, which imply that
\begin{equation}\label{pfth2-5}
I_2=O_p(\frac 1{\sqrt n})\quad\text{and}\quad I_3=O_p(\frac 1{\sqrt
n}).
\end{equation}
By (\ref{pfth2-1}), (\ref{pfth2-3})--(\ref{pfth2-5}), we have
\begin{equation}
\label{pfth2-5a} \frac 1{\sqrt n}l_{n1}(\alpha,\beta,\gamma)=\frac
1{\sqrt n}\sum_{i=1}^n\tilde Z_i+O_p(1/\sqrt n).
\end{equation}
Using
\[
\left\{\begin{array}{ll}
&\frac{\partial}{\partial u}C(u,v;\rho_i)=\Phi\Big(\frac{\Phi^-(v)-\rho_i\Phi^-(u)}{\sqrt{1-\rho_i^2}}\Big):=C_1(u,v;\rho_i)\\
&\frac{\partial}{\partial
v}C(u,v;\rho_i)=\Phi\Big(\frac{\Phi^-(u)-\rho_i\Phi^-(v)}{\sqrt{1-\rho_i^2}}\Big):=C_2(u,v;\rho_i),
\end{array}\right.\]
we have
\begin{equation}\label{pfth2-6}\begin{array}{ll}
&\int_0^1(v-\frac 12)c(u,v;\rho_i)\,dv\\
=&\int_0^1(v-\frac 12)\,C_1(u,dv;\rho_i)\\
=&\frac 12C_1(u,1;\rho_i)+\frac 12C_1(u,0;\rho_i)-\int_0^1C_1(u,v;\rho_i)\,dv\\
=&u-\frac 12-\int_0^u\Phi\Big(\frac{\Phi^-(v)-\Phi^-(u)+(1-\rho_i)\Phi^-(u)}{\sqrt{1-\rho_i^2}}\Big)\,dv\\
&+\int_u^1\left(1-\Phi\Big(\frac{\Phi^-(v)-\Phi^-(u)+(1-\rho_i)\Phi^-(u)}{\sqrt{1-\rho_i^2}}\Big)\right)\,dv\\
=&u-\frac 12-\int_{-\infty}^0\Phi\Big(v+\frac{1-\rho_i}{\sqrt{1-\rho_i^2}}\Phi^-(u)\Big)\phi\Big(v\sqrt{1-\rho_i^2}+\Phi^-(u)\Big)\sqrt{1-\rho_i^2}\,dv\\
&+\int_{0}^\infty\left(1-\Phi\Big(v+\frac{1-\rho_i}{\sqrt{1-\rho_i^2}}\Phi^-(u)\Big)\right)\phi\Big(v\sqrt{1-\rho_i^2}+\Phi^-(u)\Big)\sqrt{1-\rho_i^2}\,dv\\
=&u-\frac 12-\sqrt{1-\rho_i^2}\left(\int_{-\infty}^0\Phi(v)\phi(\Phi^-(u))\,dv+O(\sqrt{1-\rho_i})\right)\\
&+\sqrt{1-\rho_i^2}\Big(\int_0^{\infty}(1-\Phi(v))\phi(\Phi^-(u))\,dv+O(\sqrt{1-\rho_i})\Big)\\
=&u-\frac 12-\sqrt{1-\rho_i^2}\left(\frac{1}{\sqrt{2\pi}}\phi(\Phi^-(u))+O(\sqrt{1-\rho_i})\right)\\
&+\sqrt{1-\rho_i^2}\left(\frac 1{\sqrt{2\pi}}\phi(\Phi^-(u))+O(\sqrt{1-\rho_i})\right)\\
=&u-\frac 12+O(1/\log n)
\end{array}\end{equation}
and
\begin{equation}
\label{pfth2-6a}
\begin{array}{ll}
&\int_0^1(v-\frac 12)^2c(u,v;\rho_i)\,dv\\
=&\frac 14-2\int_0^1(v-\frac 12)C_1(u,v;\rho_i)\,dv\\
=&(u-\frac 12)^2-2\int_0^u(v-\frac 12)\Phi\Big(\frac{\Phi^-(v)-\Phi^-(u)+(1-\rho_i)\Phi^-(u)}{\sqrt{1-\rho_i^2}}\Big)\,dv\\
&+2\int_u^1(v-\frac 12)\left(1-\Phi\Big(\frac{\Phi^-(v)-\Phi^-(u)+(1-\rho_i)\Phi^-(u)}{\sqrt{1-\rho_i^2}}\Big)\right)\,dv\\
=&(u-\frac 12)^2-2\sqrt{1-\rho_i^2}\left(\int_{-\infty}^0(v-\frac 12)\Phi(v)\phi(\Phi^-(u))\,dv+O(\sqrt{1-\rho_i})\right)\\
&+2\sqrt{1-\rho_i^2}\Big(\int_0^{\infty}(v-\frac 12)(1-\Phi(v))\phi(\Phi^-(u))\,dv+O(\sqrt{1-\rho_i})\Big)\\
=&(u-\frac 12)^2-2\sqrt{1-\rho_i^2}\left(-(\frac 1{2\sqrt{2\pi}}+\frac 14)\phi(\Phi^-(u))+O(\sqrt{1-\rho_i})\right)\\
&+2\sqrt{1-\rho_i^2}\Big((\frac 14-\frac 1{2\sqrt{2\pi}})\phi(\Phi^-(u))+O(\sqrt{1-\rho_i})\Big)\\
=&(u-\frac 12)^2+\phi(\Phi^-(u))\sqrt{1-\rho_i^2}+O(1/\log n).
\end{array}\end{equation}
By (\ref{pfth2-6}), (\ref{pfth2-6a}), $C(u,v;1)=u\wedge v$ and
\begin{equation}\label{pfth2-6aa}
\frac{d}{d\rho}C(u,v;\rho)=\frac
1{2\pi\sqrt{1-\rho^2}}\exp\left(-\frac{(\Phi^-(u))^2-2\rho\Phi^-(u)\Phi^-(v)+(\Phi^-(v))^2}{2(1-\rho^2)}\right)
\end{equation}
(see Plackett (1954)),
 we have
\begin{equation}
\label{pfth2-7}
\begin{array}{ll}
&\E\tilde Z_{i,1}^2=\E\tilde Z_{i,2}^2\\
=&\frac 1{(n+1)^2}\sum_{j=1}^n\sum_{k=1}^n\int_0^1\int_0^1\int_0^1\int_0^1(u_1\wedge u_2-u_1u_2)(v_1-\frac 12)(v_2-\frac 12)\times\\
& c(u_1,v_1;\rho_j)c(u_2,v_2;\rho_k)\,dv_1dv_2du_1du_2\\
=&\int_0^1\int_0^1(u_1\wedge u_2-u_1u_2)(u_1-\frac 12)(u_2-\frac 12)\,du_1du_2+O(1/\log n)\\
=&\frac 1{720}+O(1/\log n),
\end{array}
\end{equation}
\begin{equation}\label{pfth2-8}\begin{array}{ll}
&\E\tilde Z_{i,3}^2\\
=&\int_0^1\int_0^1(u-\frac 12)^2(v-\frac 12)^2c(u,v;\rho_i)\,dvdu-\Big(\frac 1{2\pi}\arcsin(\frac{\rho_i}2)\Big)^2\\
=&\int_0^1(u-\frac 12)^2\Big((u-\frac 12)^2+\phi(\Phi^-(u))\sqrt{1-\rho_i^2}\Big)\,du+O(1/\log n)-\Big(\frac 1{2\pi}\arcsin(\frac 12)\Big)^2\\
=&\frac 1{80}+\frac{\sqrt 2\sqrt{m(i/n)}}{\sqrt{\log n}}\int_0^1(u-\frac 12)^2\phi(\Phi^-(u))\,du-\frac 1{144}+O(1/\log n)\\
=&\frac 1{180}+\frac{\sqrt 2\sqrt{m(i/n)}}{\sqrt{\log
n}}\int_0^1(u-\frac 12)^2\phi(\Phi^-(u))\,du+O(1/\log n),
\end{array}\end{equation}
\begin{equation}\label{pfth2-9}\begin{array}{ll}
&\E(\tilde Z_{i,1}\tilde Z_{i,2})\\
=&\frac{1}{(n+1)^2}\sum_{j=1}^n\sum_{k=1}^n\int_0^1\int_0^1\int_0^1\int_0^1(C(u_1,v_2;\rho_i)-u_1v_2)(v_1-\frac 12)(u_2-\frac 12)\times\\
&c(u_1,v_1;\rho_j)c(u_2,v_2;\rho_k)\,dv_1du_2du_1dv_2\\
=&\int_0^1\int_0^1(u_1\wedge v_2-u_1v_2)(u_1-\frac 12)(v_2-\frac 12)\,du_1dv_2+O(1/\log n)\\
=&\frac 1{720}+O(1/\log n)
\end{array}\end{equation}
and
\begin{equation}
\label{pfth2-10}\begin{array}{ll}
&\E(\tilde Z_{i,1}\tilde Z_{i,3})=\E(\tilde Z_{i,2}\tilde Z_{i,3})\\
=&\frac 1{n+1}\sum_{j=1}^n\int_0^1\int_0^1\int_0^1\int_0^1(I(u_2\le u_1)-u_1)(v_1-\frac 12)(u_2-\frac 12)(v_2-\frac 12)\times\\
&c(u_1,v_1;\rho_j)c(u_2,v_2;\rho_i)\,dv_1dv_2du_1du_2\\
=&\int_0^1\int_0^1(I(u_2\le u_1)-u_1)(u_2-\frac 12)(u_1-\frac 12)(u_2-\frac 12)\,du_1du_2+O(1/\log n)\\
=&-\frac 1{360}+O(1/\log n).
\end{array}\end{equation}
Hence, it follows from (\ref{pfth2-7})--(\ref{pfth2-10}) that
\begin{equation}\label{pfth2-11}
\frac 1{n}\sum_{i=1}^n\E\tilde Z_i^2=\frac{\sqrt
2\int_0^1\sqrt{\alpha+\beta s^{\gamma}}\,ds}{\sqrt {\log
n}}\int_0^1\left(u-\frac 12\right)^2\phi(\Phi^-(u))\,du+O(1/\log n).
\end{equation}

It is easy to check that
\[\E\left(\frac1n\sum_{i=1}^n(\tilde Z_i^2-\E\tilde Z_i^2)\right)^2=\frac 1{n^2}\sum_{i=1}^n\E(\tilde Z_i^2-\E\tilde Z_i^2)^2=O(1/n),\]
which, combining  with (\ref{pfth2-11}), implies that
\begin{equation}
\label{pfth2-11a} \sum_{i=1}^n\left(\frac{(\log n)^{1/4}}{\sqrt
n}\tilde Z_i\right)^2\overset{p}{\to} \sqrt
2\left(\int_0^1\sqrt{\alpha+\beta
s^{\gamma}}\,ds\right)\left(\int_0^1(u-\frac
12)^2\phi(\Phi^-(u))\,du\right).
\end{equation}
Obviously we have
\begin{equation}
\label{pfth2-11b} \max_{1\le i\le n}\left|\frac{(\log
n)^{1/4}}{\sqrt n}\tilde
Z_i\right|\overset{p}{\to}0\quad\text{and}\quad \E\left(\max_{1\le
i\le n}\frac{(\log n)^{1/2}}n\tilde Z_i^2\right)=o(1).
\end{equation}
Hence, it follows from (\ref{pfth2-5a}), (\ref{pfth2-11a}),
(\ref{pfth2-11b}) and Theorem 3.2 of Hall and Heyde (1980) that
\begin{equation}
\label{pfth2-12} \frac{(\log n)^{1/4}}{\sqrt
n}l_{n1}(\alpha,\beta,\gamma) \to N\left(0,\sqrt
2\Big(\int_0^1\sqrt{\alpha+\beta
s^{\gamma}}\,ds\Big)\Big(\int_0^1(u-\frac
12)^2\phi(\Phi^-(u))\,du\Big)\right).
\end{equation}

Note that the above limit has a nonstandard rate, which can be
explained as follows. When $\rho_i=1$ for $i=1,\cdots,n$, we have
\[l_{n1}(\alpha,\beta,\gamma)=\sum_{i=1}^n\left((\hat U_n(U_i)-\frac 12)^2-\frac {1}{12}\right)=\sum_{i=1}^n \left((\frac i{n+1}-\frac 12)^2-\frac {1}{12}\right),\]
which  becomes a constant. However, $l_{n2}(\alpha,\beta,\gamma)$
and $l_{n3}(\alpha,\beta,\gamma)$ are non-degenerate due to the
involved factors $(i/n)^{\gamma}$ and $(i/n)^{\gamma}\log(i/n)$.
That is, deriving the asymptotic limit of
$l_{n1}(\alpha,\beta,\gamma)$ needs finer expansions than the other
two quantities. Below we show the asymptotic limits for both
$l_{n2}$ and $l_{n3}$ have the standard rate $1/\sqrt n$.

Define
\[\begin{array}{ll}
\tilde Z_i^*&=\frac 1{(n+1)}\sum_{j=1}^n(\frac jn)^{\gamma}\int_0^1\int_0^1(I(U_i\le u)-u)(v-\frac 12)c(u,v;\rho_j)\,dudv\\
&\quad+\frac 1{(n+1)}\sum_{j=1}^n(\frac jn)^{\gamma}\int_0^1\int_0^1(I(V_i\le v)-v)(u-\frac 12)c(u,v;\rho_j)\,dudv\\
&\quad+\Big(Z_i-\frac 1{2\pi}\arcsin(\frac{\rho_i}2)\Big)(\frac in)^{\gamma}\\
&=\tilde Z_{i,1}^*+\tilde Z_{i,2}^*+\tilde Z_{i,3}^*
\end{array}\]
and
\[\begin{array}{ll}
\tilde Z_i^{**}&=\frac 1{(n+1)}\sum_{j=1}^n(\frac jn)^{\gamma}\log(\frac jn)\int_0^1\int_0^1(I(U_i\le u)-u)(v-\frac 12)c(u,v;\rho_j)\,dudv\\
&\quad+\frac 1{(n+1)}\sum_{j=1}^n(\frac jn)^{\gamma}\log(\frac jn)\int_0^1\int_0^1(I(V_i\le v)-v)(u-\frac 12)c(u,v;\rho_j)\,dudv\\
&\quad+\Big(Z_i-\frac 1{2\pi}\arcsin(\frac{\rho_i}2)\Big)(\frac in)^{\gamma}\log(\frac in)\\
&=\tilde Z_{i,1}^{**}+\tilde Z_{i,2}^{**}+\tilde Z_{i,3}^{**}.
\end{array}\]
Similar to the proof of (\ref{pfth2-5a}), we can show that
\begin{equation}\label{pfth2-12a}
\left\{\begin{array}{ll}
&\frac 1{\sqrt n}l_{n2}(\alpha,\beta,\gamma)=\frac 1{\sqrt n}\sum_{i=1}^n\tilde Z_i^*+O_p(1/\sqrt n)\\
&\frac 1{\sqrt n}l_{n3}(\alpha,\beta,\gamma)=\frac 1{\sqrt
n}\sum_{i=1}^n\tilde Z_i^{**}+O_p(1/\sqrt n).
\end{array}\right.\end{equation}
Like the proofs of (\ref{pfth2-7})--(\ref{pfth2-10}), we can show
that
\begin{equation*}
\label{pfth2-13}
\begin{array}{ll}
&\E\tilde Z_{i,1}^{*2}=\E\tilde Z_{i,2}^{*2}\\
=&(\frac 1n\sum_{j=1}^n(\frac jn)^{\gamma})^2\int_0^1\int_0^1(u_1\wedge u_2-u_1u_2)(u_1-\frac 12)(u_2-\frac 12)\,du_1du_2+O(1/\log n)\\
=&\frac {1}{720(1+\gamma)^2}+O(1/\log n),
\end{array}
\end{equation*}

\begin{equation*}\label{pfth2-14}
\left\{\begin{array}{ll}
&\E\tilde Z_{i,3}^{*2}=\frac 1{180}(\frac in)^{2\gamma}+O(1/\sqrt{\log n}),\quad \E(\tilde Z_{i,1}^*\tilde Z_{i,2}^*)=\frac 1{720(1+\gamma)^2}+O(1/\log n)\\
&\E(\tilde Z_{i,1}^*\tilde Z_{i,3}^*)=\E(\tilde Z_{i,2}^*\tilde
Z_{i,3}^*)=-\frac{1}{360(1+\gamma)}(\frac
in)^{\gamma}+O(1/\sqrt{\log n}),
\end{array}
\right.\end{equation*}

\begin{equation*}
\label{pfth2-15}
\begin{array}{ll}
&\E\tilde Z_{i,1}^{**2}=\E\tilde Z_{i,2}^{**2}\\
=&\Big(\frac 1n\sum_{j=1}^n(\frac jn)^{\gamma}\log(\frac  jn)\Big)^2\int_0^1\int_0^1(u_1\wedge u_2-u_1u_2)(u_1-\frac 12)(u_2-\frac 12)\,du_1du_2+O(1/\log n)\\
=&\frac 1{720(1+\gamma)^4}+O(1/\log n)
\end{array}
\end{equation*}
and
\[\left\{\begin{array}{ll}
&\E\tilde Z_{i,3}^{**2}=\frac 1{180}(\frac in)^{2\gamma}\log^2(\frac in)+O(1/\sqrt{\log n}),\quad \E(\tilde Z_{i,1}^{**}\tilde Z_{i,2}^{**})=\frac 1{720(1+\gamma)^4}+O(1/\log n)\\
&\E(\tilde Z_{i,1}^{**}\tilde Z_{i,3}^{**})=\E(\tilde
Z_{i,2}^{**}\tilde Z_{i,3}^{**})=\frac 1{360(1+\gamma)^2}(\frac
in)^{\gamma}\log(\frac in)+O(1/\sqrt{\log n}),
\end{array}\right.\]
which imply that
\begin{equation}
\label{pfth2-16} \left\{\begin{array}{ll}
&\E(\frac 1{\sqrt n}\sum_{i=1}^n\tilde Z_i^*)^2=\frac 1n\sum_{i=1}^n\E(\tilde Z_{i,1}^*+\tilde Z_{i,2}^*+\tilde Z_{i,3}^*)^2\to\frac 1{180(1+2\gamma)}-\frac 1{180(1+\gamma)^2}\\
&\E(\frac 1{\sqrt n}\sum_{i=1}^n\tilde Z_i^{**})^2=\frac
1n\sum_{i=1}^n\E(\tilde Z_{i,1}^{**}+\tilde Z_{i,2}^{**}+\tilde
Z_{i,3}^{**})^2\to\frac{1}{90(1+2\gamma)^3}-\frac
1{180(1+\gamma)^4}.
\end{array}\right.\end{equation}
Like the proof of (\ref{pfth2-12}), by using (\ref{pfth2-16}), we
can show that
\[\frac 1{\sqrt n}l_{n2}(\alpha,\beta,\gamma)\overset{d}{\to} N\left(0, \frac 1{180(1+2\gamma)}-\frac 1{180(1+\gamma)^2}\right)\]
and
\[\frac 1{\sqrt n}l_{n3}(\alpha,\beta,\gamma)\overset{d}{\to} N\left(0, \frac 1{90(1+2\gamma)^3}-\frac 1{180(1+\gamma)^4}\right).\]
Some further tedious calculations show that
\[\begin{array}{ll}
&\frac{(\log n)^{1/4}}{n}\sum_{i=1}^n\E(\tilde Z_i\tilde Z_i^*)\\
=&\frac{(\log n)^{1/4}}{n}\sum_{i=1}^n\sum_{j=1}^3\sum_{k=1}^3\E(\tilde Z_{i,j}\tilde Z_{i,k}^*)\\
=&\frac{(\log n)^{1/4}}{n}\sum_{i=1}^n\left(\frac 1{720(1+\gamma)}+\frac 1{720(1+\gamma)}-\frac 1{360}(\frac in)^{\gamma}\right.\\
&\left. +\frac{1}{720(1+\gamma)}+\frac1{720(1+\gamma)}-\frac 1{360}(\frac in)^{\gamma}-\frac 1{360(1+\gamma)}-\frac 1{360(1+\gamma)}+\frac 1{180}(\frac in)^{\gamma}+O(1/\sqrt{\log n})\right)\\
=&o(1),
\end{array}\]
 \[\begin{array}{ll}
&\frac{(\log n)^{1/4}}{n}\sum_{i=1}^n\E(\tilde Z_i\tilde Z_i^{**})\\
=&\frac{(\log n)^{1/4}}{n}\sum_{i=1}^n\sum_{j=1}^3\sum_{k=1}^3\E(\tilde Z_{i,j}\tilde Z_{i,k}^{**})\\
=&\frac{(\log n)^{1/4}}{n}\sum_{i=1}^n\Big(-\frac 1{720(1+\gamma)^2}-\frac 1{720(1+\gamma)^2}-\frac 1{360}(\frac in)^{\gamma}\log(\frac in)\\
&-\frac{1}{720(1+\gamma)^2}-\frac1{720(1+\gamma)^2}-\frac 1{360}(\frac in)^{\gamma}\log(\frac in)+\frac 1{360(1+\gamma)^2}+\frac 1{360(1+\gamma)^2}\\
&+\frac 1{180}(\frac in)^{\gamma}\log(\frac in)+O(1/\sqrt{\log n})\Big)\\
=&o(1)
\end{array}\]
and
\[\begin{array}{ll}
&\frac1{n}\sum_{i=1}^n\E(\tilde Z_i^*\tilde Z_i^{**})\\
=&\frac 1n\sum_{i=1}^n\sum_{j=1}^3\sum_{k=1}^3\E(\tilde Z_{i,j}^*\tilde Z_{i,k}^{**})\\
=&\frac 1n\sum_{i=1}^n\Big(-\frac 1{720(1+\gamma)^3}-\frac 1{720(1+\gamma)^3}-\frac 1{360(1+\gamma)}(\frac in)^{\gamma}\log(\frac in)\\
&-\frac 1{720(1+\gamma)^3}-\frac 1{720(1+\gamma)^3}-\frac 1{360(1+\gamma)}(\frac in)\log(\frac in)+\frac 1{360(1+\gamma)^2}(\frac in)^{\gamma}+\frac 1{360(1+\gamma)^2}(\frac in)^{\gamma}\\
&+\frac 1{180}(\frac in)^{2\gamma}\log(\frac in)+O(1/\sqrt{\log n})\Big)\\
=&\frac{1}{180(1+\gamma)^3}-\frac 1{180(1+2\gamma)^2}+o(1).
\end{array}\]
Hence, by Cram\'er device, we can show that
\begin{equation}\label{pfth2-17}
\left(\frac{(\log n)^{1/4}}{\sqrt n}l_{n1}(\alpha,\beta,\gamma), \frac 1{\sqrt n}l_{n2}(\alpha,\beta,\gamma), \frac 1{\sqrt n}l_{n3}(\alpha,\beta,\gamma)\right)^T\\
\overset{d}{\to} N(0,\Sigma).
\end{equation}

It is straightforward to check that
\begin{equation}
\label{pfth2-18} \left\{\begin{array}{ll}
&\frac{\log n}n\frac{\partial l_{n1}(\alpha,\beta,\gamma)}{\partial\alpha}\to \frac{\sqrt 3}{6\pi}, \frac{\log n}n\frac{\partial l_{n1}(\alpha,\beta,\gamma)}{\partial\beta}\to\frac{\sqrt 3}{6\pi(1+\gamma)}, \\
&\frac{\log n}n\frac{\partial l_{n1}(\alpha,\beta,\gamma)}{\partial\gamma}\to-\frac {\sqrt 3\beta}{6\pi(1+\gamma)^2},\frac{\log n}n\frac{\partial l_{n2}(\alpha,\beta,\gamma)}{\partial\alpha}\to\frac{\sqrt{3}}{6\pi(1+\gamma)}, \\
&\frac{\log n}n\frac{\partial l_{n2}(\alpha,\beta,\gamma)}{\partial\beta}\to\frac{\sqrt 3}{6\pi(1+2\gamma)}, \frac{\log n}n\frac{\partial l_{n2}(\alpha,\beta,\gamma)}{\partial\gamma}\to-\frac {\sqrt 3\beta}{6\pi(1+2\gamma)^2},\\
&\frac{\log n}n\frac{\partial l_{n3}(\alpha,\beta,\gamma)}{\partial\alpha}\to-\frac{\sqrt 3}{6\pi(1+\gamma)^2}, \frac{\log n}n\frac{\partial l_{n3}(\alpha,\beta,\gamma)}{\partial\beta}\to-\frac{\sqrt 3}{6\pi(1+2\gamma)^2}, \\
&\frac{\log n}n\frac{\partial l_{n3}(\alpha,\beta,\gamma)}{\partial\gamma}\to\frac {\sqrt 3\beta}{3\pi(1+2\gamma)^3 }.\\
\end{array}\right.\end{equation}
Hence, the theorem follows from (\ref{pfth2-17}), (\ref{pfth2-18})
and Taylor expansion.
\end{proof}

\begin{proof}[Proof of Theorem \ref{th2a}]
As in the proof of Theorem \ref{th2}, we define
\[\begin{array}{ll}
\bar Z_i&=\frac 1{(n+1)}\sum_{j=1}^n\int_0^1\int_0^1\frac{\Phi^-(v)}{\phi(\Phi^-(u))}(I(U_i\le u)-u)c(u,v;\rho_j)\,dudv\\
&\quad+\frac 1{(n+1)}\sum_{j=1}^n\int_0^1\int_0^1\frac{\Phi^-(u)}{\phi(\Phi^-(v))}(I(V_i\le v)-v)c(u,v;\rho_j)\,dudv\\
&\quad+\Big(\Phi^-(U_i)\Phi^-(V_i)-\rho_i\Big)\\
&=\bar Z_{i,1}+\bar Z_{i,2}+\bar Z_{i,3},
\end{array}\]
\[\begin{array}{ll}
\bar Z_i^*&=\frac 1{(n+1)}\sum_{j=1}^n(\frac jn)^{\gamma}\int_0^1\int_0^1\frac{\Phi^-(v)}{\phi(\Phi^-(u))}(I(U_i\le u)-u)c(u,v;\rho_j)\,dudv\\
&\quad+\frac 1{(n+1)}\sum_{j=1}^n(\frac jn)^{\gamma}\int_0^1\int_0^1\frac{\Phi^-(u)}{\phi(\Phi^-(v))}(I(V_i\le v)-v)c(u,v;\rho_j)\,dudv\\
&\quad+\Big(\Phi^-(U_i)\Phi^-(V_i)-\rho_i\Big)(\frac in)^{\gamma}\\
&=\bar Z_{i,1}^*+\bar Z_{i,2}^*+\bar Z_{i,3}^*
\end{array}\]
and
\[\begin{array}{ll}
\bar Z_i^{**}&=\frac 1{(n+1)}\sum_{j=1}^n(\frac jn)^{\gamma}\log(\frac jn)\int_0^1\int_0^1\frac{\Phi^-(v)}{\phi(\Phi^-(u))}(I(U_i\le u)-u)c(u,v;\rho_j)\,dudv\\
&\quad+\frac 1{(n+1)}\sum_{j=1}^n(\frac jn)^{\gamma}\log(\frac jn)\int_0^1\int_0^1\frac{\Phi^-(u)}{\phi(\Phi^-(v))}(I(V_i\le v)-v)c(u,v;\rho_j)\,dudv\\
&\quad+\Big(\Phi^-(U_i)\Phi^-(V_i)-\rho_i\Big)(\frac in)^{\gamma}\log(\frac in)\\
&=\bar Z_{i,1}^{**}+\bar Z_{i,2}^{**}+\bar Z_{i,3}^{**}.
\end{array}\]
Since
\[\int_0^1\Phi^-(v)c(u,v;\rho_i)\,dv=\rho_i\Phi^-(u)\quad\text{and}\quad \int_0^1\Phi^-(u)c(u,v;\rho_i)\,du=\rho_i\Phi^-(v),\]
we have
\[\left\{\begin{array}{ll}
&\bar Z_{i,1}=\left(\frac 1{n+1}\sum_{j=1}^n\rho_j\right)\int_0^1\frac{\Phi^-(u)}{\phi(\Phi^-(u))}(I(U_i\le u)-u)\,du, \\
&\bar Z_{i,2}=\left(\frac 1{n+1}\sum_{j=1}^n\rho_j\right)\int_0^1\frac{\Phi^-(v)}{\phi(\Phi^-(v))}(I(V_i\le v)-v)\,dv,\\
&\bar Z_{i,1}^*=\left(\frac 1{n+1}\sum_{j=1}^n\rho_j(\frac jn)^{\gamma}\right)\int_0^1\frac{\Phi^-(u)}{\phi(\Phi^-(u))}(I(U_i\le u)-u)\,du, \\
&\bar Z_{i,2}^*=\left(\frac 1{n+1}\sum_{j=1}^n\rho_j(\frac jn)^{\gamma}\right)\int_0^1\frac{\Phi^-(v)}{\phi(\Phi^-(v))}(I(V_i\le v)-v)\,dv,\\
&\bar Z_{i,1}^{**}=\left(\frac 1{n+1}\sum_{j=1}^n\rho_j(\frac jn)^{\gamma}\log(\frac jn)\right)\int_0^1\frac{\Phi^-(u)}{\phi(\Phi^-(u))}(I(U_i\le u)-u)\,du, \\
&\bar Z_{i,2}^{**}=\left(\frac 1{n+1}\sum_{j=1}^n\rho_j(\frac
jn)^{\gamma}\log(\frac
jn)\right)\int_0^1\frac{\Phi^-(v)}{\phi(\Phi^-(v))}(I(V_i\le
v)-v)\,dv.
\end{array}\right.\]
It is straightforward to check that
\begin{equation}\label{pfth2a-1}
\begin{array}{ll}
\E\bar Z_{i,1}^2&=\left(\frac 1{n+1}\sum_{j=1}^n\rho_j\right)^2\int_0^1\int_0^1\frac{\Phi^-(u_1)\Phi^-(u_2)}{\phi(\Phi^-(u_1))\phi(\Phi^-(u_2))}(u_1\wedge u_2-u_1u_2)\,du_1du_2\\
&=\left(\frac 1{n+1}\sum_{j=1}^n\rho_j\right)^22\int_0^1\int_0^{u_1}\frac{\Phi^-(u_1)\Phi^-(u_2)}{\phi(\Phi^-(u_1))\phi(\Phi^-(u_2))}u_2(1-u_1)\,du_2du_1\\
&=\left(\frac 1{n+1}\sum_{j=1}^n\rho_j\right)^2\int_0^1\frac{\Phi^-(u_1)}{\phi(\Phi^-(u_1))}(1-u_1)\int_{-\infty}^{\Phi^-(u_1)}\Phi(u_2)\,du_2^2du_1\\
&=\left(\frac 1{n+1}\sum_{j=1}^n\rho_j\right)^2\int_0^1\frac{\Phi^-(u_1)}{\phi(\Phi^-(u_1))}(1-u_1)\{(\Phi^-(u_1))^2u_1+\int_{-\infty}^{\Phi^-(u_1)}u_2\,d\phi(u_2)\}\,du_1\\
&=\left(\frac 1{n+1}\sum_{j=1}^n\rho_j\right)^2\int_0^1\frac{\Phi^-(u_1)}{\phi(\Phi^-(u_1))}(1-u_1)\{(\Phi^-(u_1))^2u_1+\Phi^-(u_1)\phi(\Phi^-(u_1))-u_1\}\,du_1\\
&=\left(\frac 1{n+1}\sum_{j=1}^n\rho_j\right)^2\int_{-\infty}^{\infty}u(1-\Phi(u))\{u^2\Phi(u)+u\phi(u)-\Phi(u)\}\,du\\
&=-\left(\frac 1{n+1}\sum_{j=1}^n\rho_j\right)^2\int_{-\infty}^{\infty}u(1-\Phi(u))\,d\phi(u)\\
&=\left(\frac 1{n+1}\sum_{j=1}^n\rho_j\right)^2\int_{-\infty}^{\infty}\phi(u)\{1-\Phi(u)-u\phi(u)\}\,du\\
&=\frac 12 \left(\frac 1{n+1}\sum_{j=1}^n\rho_j\right)^2\\
\end{array}\end{equation}
by noting that $u(1-\Phi(u))\Phi(u)$, $u^3(1-\Phi(u))\Phi(u)$ and
$u\phi^2(u)$ are odd functions,
\begin{equation}\label{pfth2a-2}
\E\bar Z_{i,2}^2=\frac 12\left(\frac
1{n+1}\sum_{j=1}^n\rho_j\right)^2 \quad\text{and}\quad \E\bar
Z_{i,3}^2=1+2\rho_i^2-\rho_i^2=1+\rho_i^2.
\end{equation}
By (\ref{pfth2-6aa}), we have
\[\int_0^1\int_0^1\frac{\Phi^-(u)\Phi^-(v)}{\phi(\Phi^-(u))\phi(\Phi^-(v))}\frac{dC(u,v;\rho)}{d\rho}\,dudv=\rho\]
for any $\rho\in (-1, 1)$. Taking derivative with respect to $\rho$
at both sides, we have
\[\int_0^1\int_0^1\frac{\Phi^-(u)\Phi^-(v)}{\phi(\Phi^-(u))\phi(\Phi^-(v))}\frac{d^2C(u,v;\rho)}{d\rho^2}\,dudv=1\]
for any $\rho\in (-1, 1)$. Therefore
\[\begin{array}{ll}
\frac 12&=\int_0^1\int_0^1\frac{\Phi^-(u)\Phi^-(v)}{\phi(\Phi^-(u))\phi(\Phi^-(v))}(u\wedge v-uv)\,dudv\\
&=\int_0^1\int_0^1\frac{\Phi^-(u)\Phi^-(v)}{\phi(\Phi^-(u))\phi(\Phi^-(v))}(C(u,v;1)-uv)\,dudv\\
&=\int_0^1\int_0^1\frac{\Phi^-(u)\Phi^-(v)}{\phi(\Phi^-(u))\phi(\Phi^-(v))}(C(u,v;\rho_i)-uv)\,dudv\\
&\quad+\rho_i(1-\rho_i)+\frac 12(1-\rho_i)^2+o(1/\log^2n),
\end{array}\]
which gives
\begin{equation}\label{pfth2a-4}
\begin{array}{ll}
\E(\bar Z_{i,1}\bar Z_{i,2})&=\left(\frac 1{n+1}\sum_{j=1}^n\rho_j\right)^2\int_0^1\int_0^1\frac{\Phi^-(u)\Phi^-(v)}{\phi(\Phi^-(u))\phi(\Phi^-(v))}(C(u,v;\rho_i)-uv)\,dudv\\
&=\left(\frac 1{n+1}\sum_{j=1}^n\rho_j\right)^2\left(\frac 12-\rho_i(1-\rho_i)-\frac 12(1-\rho_i)^2\right)+o(1/\log^2n)\\
&=\left(\frac 1{n+1}\sum_{j=1}^n\rho_j\right)^2\frac{\rho_i^2}2
+o(1/\log^2 n).
\end{array}\end{equation}
Since
\[\begin{array}{ll}
\E\left(I(U_i\le u)\Phi^-(U_i)\Phi^-(V_i)\right)&=\E\Big(I(U_i\le u)\Phi^-(U_i)\E(\Phi^-(V_i)|\Phi^-(U_i))\Big)\\
&=\E(I(U_i\le u)\Phi^-(U_i)\rho_i\Phi^-(U_i))\\
&=\rho_i\int_{-\infty}^{\Phi^-(u)}v^2\phi(v)\,dv\\
&=-\rho_i\int_{-\infty}^{\Phi^-(u)}v\,d\phi(v)\\
&=-\rho_{i}\left(\Phi^-(u)\phi(\Phi^-(u))-u\right),
\end{array}\] we have
\begin{equation}
\label{pfth2a-5}
\begin{array}{ll}
&\E(\bar Z_{i,1}\bar Z_{i,3})=\E(\bar Z_{i,2}\bar Z_{i,3})\\
=&\left(\frac 1{n+1}\sum_{j=1}^n\rho_j\right)\int_0^1\frac{\Phi^-(u)}{\phi(\Phi^-(u))}\rho_i\left(-\Phi^-(u)\phi(\Phi^-(u))+u-u\right)\,du\\
=&-\left(\frac 1{n+1}\sum_{j=1}^n\rho_j\right)\rho_i.\\
\end{array}\end{equation}
Put $\bar\rho=n^{-1}\sum_{j=1}^n\rho_j$. Then
\begin{equation}\label{pfth2a-5a}
\left\{\begin{array}{ll}
&\bar\rho=1-\frac{\alpha+\beta/(1+\gamma)}{\log n}+o(1/\log^2n)\\
&\frac 1n\sum_{j=1}^n\rho_j^2=\frac
1n\sum_{j=1}^n(\rho_j-\bar\rho)^2+\bar\rho^2=\frac{\beta^2}{\log^2
n}(\frac 1{1+2\gamma}-\frac
1{(1+\gamma)^2})+\bar\rho^2+o(1/\log^2n).
\end{array}\right.\end{equation}
It follows from (\ref{pfth2a-1})--(\ref{pfth2a-5a}) that
\begin{equation}\label{pfth2a-6}
\begin{array}{ll}
&\frac{1}{n}\sum_{i=1}^n\E\bar Z_{i}^2\\
=&(\frac 1{n+1}\sum_{j=1}^n\rho_j)^2+1+\frac 1n\sum_{i=1}^n\rho_i^2+\frac 1n\sum_{i=1}^n\rho_i^{2}(\frac 1{n+1}\sum_{j=1}^n\rho_j)^2\\
&-\frac 4n\sum_{i=1}^n\rho_i(\frac 1{n+1}\sum_{j=1}^n\rho_j)+o(1/\log^2n)\\
=&\bar\rho^2+1+\bar\rho^2+\frac 1n\sum_{j=1}^n(\rho_j-\bar\rho)^2+(\bar\rho^2+\frac 1n\sum_{j=1}^n(\rho_j-\bar\rho)^2)\bar\rho^2\\
&-4\bar\rho^2+o(1/\log^2n)\\
=&4(\frac{\alpha+\beta/(1+\gamma)}{\log
n})^2+\frac{2\beta^2}{\log^2n}(\frac1{1+2\gamma}-\frac
1{(1+\gamma)^2})+o(1/\log^2n).
\end{array}
\end{equation}
Similarly we can show that
\begin{equation}
\label{pfth2a-7}
\begin{array}{ll}
&\frac 1n\sum_{i=1}^n\E\bar Z_{i}^{*2}\\
=&\frac 1n\sum_{i=1}^n\sum_{j=1}^3\sum_{k=1}^3\E(\bar Z_{i,j}^*\bar Z_{i,k}^*)\\
=&\frac 1n\sum_{i=1}^n\left(\frac 1{2(1+\gamma)^2}+\frac 1{2(1+\gamma)^2}-\frac 1{1+\gamma}(\frac in)^{\gamma}\right.\\
&\left. +\frac 1{2(1+\gamma)^2}+\frac 1{2(1+\gamma)^2}-\frac 1{1+\gamma}(\frac in)^{\gamma}-\frac 1{1+\gamma}(\frac in)^{\gamma}-\frac 1{1+\gamma}(\frac in)^{\gamma}+2(\frac in)^{2\gamma} \right) +o(1)\\
=&\frac{2}{1+2\gamma}-\frac{2}{(1+\gamma)^2}+o(1),
\end{array}\end{equation}
\begin{equation}
\label{pfth2a-8}
\begin{array}{ll}
&\frac 1n\sum_{i=1}^n\E\bar Z_{i}^{**2}\\
=&\frac 1n\sum_{i=1}^n\sum_{j=1}^3\sum_{k=1}^3\E(\bar Z_{i,j}^{**}\bar Z_{i,k}^{**})\\
=&\frac 1n\sum_{i=1}^n\left(\frac 1{2(1+\gamma)^4}+\frac 1{2(1+\gamma)^4}+\frac 1{(1+\gamma)^2}(\frac in)^{\gamma}\log(\frac in)\right.\\
&\left. +\frac 1{2(1+\gamma)^4}+\frac 1{2(1+\gamma)^4}+\frac 1{(1+\gamma)^2}(\frac in)^{\gamma}\log(\frac in)\right.\\
&\left. +\frac 1{(1+\gamma)^2}(\frac in)^{\gamma}\log(\frac in)+\frac 1{(1+\gamma)^2}(\frac in)^{\gamma}\log(\frac in)+2(\frac in)^{2\gamma}\log^2(\frac in)\right) +o(1)\\
=&\frac{4}{(1+2\gamma)^3}-\frac{2}{(1+\gamma)^4}+o(1),
\end{array}\end{equation}
\begin{equation}\label{pfth2a-9}
\frac{\log n}{n}\sum_{i=1}^n\E(\bar Z_{i}\bar Z_i^*)=o(1),\quad
\frac{\log n}n\sum_{i=1}^n\E(\bar Z_i\bar Z_i^{**})=o(1)
\end{equation}
and
\begin{equation}\label{pfth2a-10}
\begin{array}{ll}
&\frac 1n\sum_{i=1}^n\E(\bar Z_i^*\bar Z_i^{**})\\
=&\frac 1n\sum_{i=1}^n\sum_{j=1}^3\sum_{k=1}^3\E(\bar Z_{i,j}^*\bar Z_{i,k}^{**})\\
=&\frac 1n\sum_{i=1}^n\left(-\frac 1{2(1+\gamma)^3}-\frac 1{2(1+\gamma)^3}-\frac1{1+\gamma}(\frac in)^{\gamma}\log(\frac in) \right.\\
& \left. -\frac 1{2(1+\gamma)^3}-\frac 1{2(1+\gamma)^3}-\frac 1{1+\gamma}(\frac in)^{\gamma}\log(\frac in) \right. \\
&\left. +\frac1{(1+\gamma)^2}(\frac in)^{\gamma}+\frac 1{(1+\gamma)^2}(\frac in)^{\gamma}+2(\frac in)^{2\gamma}\log(\frac in)\right) +o(1)\\
=&-\frac 2{(1+2\gamma)^2}+\frac 2{(1+\gamma)^3}+o(1).
\end{array}\end{equation}
Therefore, using (\ref{pfth2a-6})--(\ref{pfth2a-10}) and the same
arguments in proving (\ref{pfth2-17}), we can show that
\begin{equation}
\label{pfth2a-11} \left(\frac {\log n}{\sqrt
n}l_{n1}^*(\alpha,\beta,\gamma), \frac 1{\sqrt
n}l_{n2}^*(\alpha,\beta,\gamma), \frac 1{\sqrt
n}l_{n3}^*(\alpha,\beta,\gamma)\right)^T\overset{d}{\to} N(0,
\Sigma^*).
\end{equation}
It is straightforward to check that
\begin{equation}
\label{pfth2a-12} \left\{\begin{array}{ll}
&\frac{\log n}n\frac{\partial l_{n1}^*(\alpha,\beta,\gamma)}{\partial\alpha}=1, \frac{\log n}n\frac{\partial l_{n1}^*(\alpha,\beta,\gamma)}{\partial\beta}\to\frac1{1+\gamma}, \frac{\log n}n\frac{\partial l_{n1}^*(\alpha,\beta,\gamma)}{\partial\gamma}\to-\frac {\beta}{(1+\gamma)^2},\\
&\frac{\log n}n\frac{\partial l_{n2}^*(\alpha,\beta,\gamma)}{\partial\alpha}\to\frac1{1+\gamma}, \frac{\log n}n\frac{\partial l_{n2}^*(\alpha,\beta,\gamma)}{\partial\beta}\to\frac1{1+2\gamma}, \frac{\log n}n\frac{\partial l_{n2}^*(\alpha,\beta,\gamma)}{\partial\gamma}\to-\frac {\beta}{(1+2\gamma)^2},\\
&\frac{\log n}n\frac{\partial l_{n3}^*(\alpha,\beta,\gamma)}{\partial\alpha}\to-\frac{1}{(1+\gamma)^2}, \frac{\log n}n\frac{\partial l_{n3}^*(\alpha,\beta,\gamma)}{\partial\beta}\to-\frac1{(1+2\gamma)^2}, \frac{\log n}n\frac{\partial l_{n3}^*(\alpha,\beta,\gamma)}{\partial\gamma}\to\frac {2\beta}{(1+2\gamma)^3},\\
\end{array}\right.\end{equation}
Hence, the theorem follows from (\ref{pfth2a-11}), (\ref{pfth2a-12})
and Taylor expansions.
\end{proof}

\begin{proof}[Proof of Theorem \ref{th3}.]
It follows from the proof of Theorem \ref{th2} with known
$\gamma=1$.
\end{proof}

\begin{proof}[Proof of Theorem \ref{th3a}.]
It follows from the proof of Theorem \ref{th2a} with known
$\gamma=1$.
\end{proof}

\begin{proof}[Proof of Theorem \ref{th4}.]
Note that
\begin{equation}\label{pfth4-1}
(\log n)Q''(s)\to -\frac{\sqrt 3m''(s)}{6\pi}\quad\text{and}\quad
\cos(2\pi Q(s))\to\frac{\sqrt 3}2\end{equation} and
\begin{equation}
\label{pfth4-1a}
\begin{array}{ll}
&\E\Big((F_1(X_i)-\frac 12)(F_2(Y_i)-\frac 12)-\frac 1{2\pi}\arcsin(\frac{\rho_i}2)\Big)^2\\
=&\E\left((F_1(X_i)-\frac 12)^2(F_2(Y_i)-\frac 12)^2\right)-\left(\frac 1{2\pi}\arcsin(\frac{\rho_i}2)\right)^2\\
\to&\frac 1{80}-(\frac {1}{12})^2=\frac {1}{180}.
\end{array}\end{equation}
It follows from (\ref{pfth4-1a}) and the standard arguments in local
linear estimation (e.g., Fan and Gijbels \cite{fan}) that
\begin{equation}\label{pfth4-2}
\sqrt{nh}\left(\hat Q(s)-Q(s)-\frac 12
Q''(s)h^{2}\int_{-1}^1t^2k(t)\,dt\right) \overset{d}{\to} N\left(0,
\frac 1{180}\int_{-1}^1k^2(t)\,dt \right).
\end{equation}
Hence it follows from  (\ref{pfth4-1}) and (\ref{pfth4-2}) that
\[\begin{array}{ll}
&\frac{\sqrt{nh}}{\log n}(\hat m(s)-m(s))\\
=&-4\pi \cos(2\pi Q(s))\sqrt{nh}(\hat Q(s)-Q(s))+o_p(1)\\
\overset{d}{\to}&N\left(\frac 12\lambda
m''(s)\int_{-1}^1t^2k(t)\,dt,~ \frac
{\pi^2}{15}\int_{-1}^1k^2(t)\,dt\right),
\end{array}\]
i.e., the theorem holds.
\end{proof}

\begin{proof}[Proof of Theorem \ref{th4a}]
It follows from standard arguments in local linear estimation.
\end{proof}

\vspace{1cm}

\noindent {\bf Acknowledgements}~~We thank two reviewers for helpful comments. Liang Peng was partly supported by the Simons Foundation. Zuoxiang Peng was supported by the National Natural
Science Foundation of China (Grant no. 11171275) and the Natural
Science Foundation Project of CQ (Grant no. cstc2012jjA00029).

\end{document}